\documentclass[12pt]{iopart}

\usepackage{iopams}
\usepackage{setstack}
\usepackage{amsfonts}
\usepackage{amssymb}
\usepackage{amscd}
\usepackage{amsopn}
\usepackage{amsthm}
\usepackage{amsfonts}
\usepackage{enumerate}
\usepackage{graphicx}
\usepackage[tight,TABTOPCAP]{subfigure}
\usepackage{bm,bbm}
\usepackage{color}
\usepackage{dsfont}
\usepackage{cite}

\newcommand{\be}{\begin{equation}}
\newcommand{\ee}{\end{equation}}
\newcommand{\bes}{\begin{equation*}}
\newcommand{\ees}{\end{equation*}}
\newcommand{\bea}{\begin{eqnarray}}
\newcommand{\eea}{\end{eqnarray}}
\newcommand{\beas}{\begin{eqnarray*}}
\newcommand{\eeas}{\end{eqnarray*}}
\newcommand{\bl}{\begin{lem}}
\newcommand{\el}{\end{lem}}
\newcommand{\la}{\lambda}
\newcommand{\al}{\alpha}

\newcommand{\f}{\frac{1}{2}}

\DeclareMathOperator{\diag}{diag}
\DeclareMathOperator{\Vol}{Vol}
\DeclareMathOperator{\Real}{Re}
\DeclareMathOperator{\Image}{Im}
\DeclareMathOperator{\diffi}{d}

\DeclareMathOperator{\sgn}{sgn}
\DeclareMathOperator{\erfc}{erfc}

\DeclareMathOperator{\Pfaff}{Pfaff}
\DeclareMathOperator{\re}{Re}
\DeclareMathOperator{\im}{Im}

\newtheorem{thm}{Theorem}[section]
\newtheorem{lem}[thm]{Lemma}
\newtheorem{defi}[thm]{Definition}

\newcommand{\bd}{\begin{defi}}
\newcommand{\ed}{\end{defi}}

\begin{document}

\title[Induced Ginibre ensembles of random matrices]
{Induced Ginibre ensemble of random matrices \\
               and quantum operations}

\author{ Jonit Fischmann$^1$, Wojciech Bruzda$^2$,  Boris A Khoruzhenko$^1$,
Hans-J{\"u}rgen Sommers$^{3}$, Karol {\.Z}yczkowski$^{2,4}$}

\address{$^1$  Queen Mary University of London, School of Mathematical Sciences,
 London E1 4NS, UK}

\address{$^2$ Institute of Physics,  Jagiellonian University,
 ul. Reymonta 4, 30-059 Krak{\'o}w, Poland}

\address{$^3$ Fakult\"at Physik, Universit\"at Duisburg-Essen, 47048 Duisburg, Germany}

\address{$^4$ Center for Theoretical Physics, Polish Academy of Sciences,
         Al. Lotnik{\'o}w 32/44, 02-668 Warszawa, Poland}

\ead{  j.fischmann@qmul.ac.uk \quad w.bruzda@uj.edu.pl \quad
  b.khoruzhenko@qmul.ac.uk \quad h.j.sommers@uni-due.de \quad karol@tatry.if.uj.edu.pl}

 \date{July 12, 2011}

\begin{abstract}
A generalisation of the Ginibre ensemble of non--Hermitian random square matrices is introduced.  The corresponding probability measure is induced by the ensemble of rectangular  Gaussian matrices via a quadratisation procedure. We derive the joint probability density of eigenvalues for such induced Ginibre ensemble  and study various spectral correlation functions
for complex and real  matrices, and analyse universal behaviour in the limit of large dimensions. In this limit the eigenvalues of the induced Ginibre ensemble cover uniformly a ring in the complex plane. The real induced Ginibre ensemble is shown to be useful to describe statistical properties of evolution operators associated with random quantum operations,
for which the dimensions of the input state and the output state do differ.
\end{abstract}


\section{Introduction}
\label{sec:1}
In 1965 Ginibre introduced a new three fold family of non--Hermitian Gaussian random matrix ensembles as an extension to the mathematical theory of Hermitian random matrices \cite{Ginibre1965}. Since then Ginibre ensembles of random matrices  have found numerous applications in various fields of physics. They can be used to describe non-unitary dynamics of open  quantum systems \cite{Ha10}, transfer matrices in mesoscopic wires \cite{Be97}, directed quantum chaos \cite{Ef97}. Real non-symmetric random matrices can be used in mathematical finances to describe correlations  between time series of prices of  various stocks \cite{BP,KDGO06}, and in physiology to characterize correlations between data representing the electric activity of brain \cite{KDI00,Se03}. The same ensemble of real Ginibre matrices describes spectral properties of evolution operators representing random quantum operations \cite{BSCSZ10} and is useful in telecommunication to characterize scattering of electromagnetic waves on random obstacles \cite{TV04}, and to describe the effects of synchronization in random networks \cite{Ti02}.

In his original paper Ginibre derived the eigenvalue distribution of matrices with i.i.d. real, complex or quaternion-real entries \cite{Ginibre1965}. These Ginibre ensembles are sometimes denoted in literature \cite{S} as  {\sf GinOE}, {\sf GinUE} and {\sf GinSE}, respectively. The letter U, O and S stands for orthogonal, unitary and symplectic symmetry class. The case {\sf GinOE} of real asymmetric matrices proved to be the hardest and Ginibre studied only the special case that all eigenvalues are real.

It took another 25 years for Lehmann and Sommers \cite{LS} and Edelman \cite{E} to derive the complete distribution of eigenvalues for the real Ginibre ensemble. Further difficulty arose in the computation of the eigenvalue correlation functions. In 2007 Akemann and Kanzieper succeeded in expressing the smooth complex correlation functions as Pfaffians \cite{AK}, whereas Sinclair presented a method for averaging over the real Ginibre ensemble in terms of Pfaffians \cite{S}. Finally, Forrester and Nagao were able to determine the real-real as well as the complex-complex correlation functions as Pfaffians using the method of skew-orthogonal polynomials \cite{FN,Forrester2008}, while Borodin and Sinclair gave the real-complex correlation in addition to  a thorough asymptotic analysis \cite{BS}. Simultaneously and independently, Sommers \cite{Sommers2007} and Sommers and Wieczorek \cite{SW} derived the complex-complex, real-real, and complex-real eigenvalue correlation functions via free-fermion diagram expansion.  A similar progress was made for the chiral real Gaussian ensemble \cite{Akemann2010b} (see also \cite{Osborn2004,Akemann2005b} for the chiral complex and quaternion-real Gaussian ensembles) and for two non-Gaussian ensembles of real asymmetric matrices \cite{Forrester2011,KSZ10}. A general review on non--Hermitian random matrices can be found in \cite{FS03}, while a recent overview on the Ginibre ensembles is provided in \cite{KS}.

A square matrix $A$ of size $N$  pertaining to the complex (real) {\sl Ginibre ensemble} consists of $N^2$ independent complex (real) random numbers drawn from a normal distribution with zero mean and a fixed variance \cite{Me04}. Normalizing to $\langle \Tr AA^{\dagger} \rangle= N^2$, the joint probability density function of the matrix entries in this ensemble is
\begin{equation}
p_{\mathrm{G}}(A) \propto  \exp\left( -\frac{\beta}{2} \Tr A^{\dagger} A\right)\; ,
\label{ginib1}
\end{equation}
where $\beta=1$ for real matrices and $\beta=2$ for complex matrices. $A^{\dagger}$ is the Hermitian conjugate of $A$ which for the real matrices is simply the transpose of $A$. Note that the same formula can be used also to define an analogous probability measure in the space of rectangular matrices $A$.  The spectral density of the Ginibre ensemble is described by the Girko circular law  \cite{Gi84} according to which the eigenvalue distribution for large $N$ is, in the leading order, uniform in a disk about the origin in the complex plane \cite{Ba97,GT07,TV08}.

The Ginibre ensemble was generalised by Feinberg and Zee \cite{FZ97,FSZ01} who studied random complex matrices with joint probability density of matrix entries
\begin{equation}
p_{\mathrm{FZ}}(A) \propto   \exp\left( -\Tr\; V( A^{\dagger}A) \right) \; ,
\label{ginibFZ}
\end{equation}
where $V( A^{\dagger}A)$ is a polynomial in $A^{\dagger}A$. It was found in \cite{FZ97,FSZ01} that in the limit of large matrix dimensions the spectral measure of (\ref{ginibFZ})  can only be supported by a single ring (or disk) on the complex plane. A mathematical proof of the ``Single Ring Theorem'' was recently provided by Guionnet et al. \cite{GKZ09} who considered random matrices $A=UDV$ with independent Haar unitary $U$ and $V$, and real diagonal $D$,  see also \cite{Wei08,Bo10} for an alternative approach. The Feinberg-Zee ensemble (\ref{ginibFZ}) fits into this scheme on making use of  the singular value decomposition in (\ref{ginibFZ}). This is because the corresponding Jacobian depends only on the singular values of $A$, see, e.g., \cite{Forrester2010b}. The matrices considered in  \cite{Wei08,Bo10} are of the form $A=UD$ with Haar unitary $U$ and diagonal $D$, however, as far as eigenvalues are concerned the two ensembles $UDV$ and $UD$ are equivalent because of the invariance of the Haar measure. Yet another approach to obtaining the spectral measure of products of random matrices was proposed in \cite{Burda2010,Burda2010a} and extended to weighted sums of unitary matrices in \cite{Jarosz2010}

In this work we study a generalisation of the Ginibre ensemble of $N\times N$ matrices, complex ($\beta=2$) or real ($\beta=1$), specified by the probability distribution with density  
\begin{equation}
p^{(\beta)}_{\mathrm{IndG}}(G)= C^{(\beta)}_L\left(\det G^{\dagger} G\right)^{\frac{\beta }{2}\,L}   \exp\left(  -\frac{\beta}{2}\Tr  G^{\dagger} G\right)
\label{ginibind}
\end{equation}
with the normalization constant
\begin{equation} \label{ginibind1}
C^{(\beta)}_L=\pi^{-\frac{\beta}{2}N^2}\Big(\frac{\beta}{2}\Big)^{\f N^2-\f NL}\prod_{j=1}^N\frac{\Gamma(\frac{\beta}{2}j)}{\Gamma\big(\frac{\beta}{2}(j+L)\big)}\; .
\end{equation}
Here $L$ is a free parameter such that $L\ge 0$. Formally, the random matrix ensemble (\ref{ginibind}) is a special case of the Feinberg-Zee ensemble (\ref{ginibFZ}) corresponding to the potential $V(t)=-\frac{\beta}{2}(t-L\, \log t)$. This specific choice of potential makes the model exactly solvable in that one is able to obtain the joint probability density function of  eigenvalues and, consequently, study the eigenvalue statistics to the same level of detail as in the original Ginibre ensemble. For example, the distribution of the real eigenvalues for real matrices and all eigenvalue correlation functions become accessible. Obtaining this information by the techniques used to study the Feinberg-Zee ensemble seems to be hardly possible.

Ensemble (\ref{ginibind}) will be called induced Ginibre ensemble as we will show in Section \ref{w:section4.2}  that a random matrix $G$ from this ensemble can be generated out of an auxiliary rectangular Gaussian matrix $X$ of size $(N+L)\times N$ and a random unitary matrix distributed according to the Haar measure on the unitary group. Thus the Gaussian measure on the space of $(N+L)\times N$ rectangular matrices is used to induce another measure in the space of square $N\times N$ matrices.  A similar construction is used to generate  ensembles of quantum states \cite{ZS01}, as the Haar measure on the group of unitary matrices of size $kN$ induces a probability measure in the space of mixed quantum states of size $N$. Random square matrices of size $N$ from the induced ensemble can also be obtained by means of quadratisation (\ref{WX}) of rectangular $(N+L)\times N$ random matrices.

The induced Ginibre ensemble of real non-symmetric matrices can be defined in an analogous way. We study spectral density and eigenvalue correlation functions for both induced ensembles of random matrices. We show that the induced Ginibre ensemble exhibits interesting behaviour in the limit of large matrix dimensions: Its eigenvalues spreading across an annulus in the complex plane, as opposed to the circular law. Secondly, on the level of correlation functions our asymptotic analysis reveals Ginibre correlations, supporting the universality conjecture.

This article is organized as follows. In Section \ref{w:section4.2} we introduce a procedure for quadratisation of rectangular matrices and we derive the induced Ginibre distribution.  In Section \ref{sec:compl} the induced Ginibre ensemble of complex matrices is investigated, its spectral density is derived and eigenvalue correlation functions  are analysed. Section  \ref{sec:real} is devoted to the induced Ginibre ensemble of real matrices. In Section \ref{sec:quant} it is demonstrated that the real induced Ginibre ensemble can be linked to evolution operators associated to generic complementary quantum operations, which modify the dimensionality of the quantum system. Proofs of technical results are relegated to the appendix.

\section{Quadratisation of rectangular matrices and induced Ginibre ensembles}
\label{w:section4.2}

Consider complex or real rectangular matrices $X$ with $M$ rows and $N$ columns. We shall assume that $M > N$, so that our rectangular matrices are 'standing'.  Let $Y$ and $Z$ denote, respectively, the upper square block of size $N\times N$ and the lower rectangular block of size $(M-N)\times N$ of the matrix $X$.  Since the standard definition of the spectrum does not work for non--square matrices we provide a unitary transformation $W\in U(M)$  intended to set the lower block $Z$  to zero :
\begin{equation}\label{WX}
W^{\dagger} X =W^{\dagger} \left[\begin{array}{c}
Y\\Z \end{array}\right]=
\left[\begin{array}{c}
G \\0 \end{array}\right] \; .
\end{equation}

One can easily find such transformations. Assuming that the matrix $X$ has rank $N$,  consider the linear span ${\cal S}$ of the column-vectors of $X$. Let $\bi{q}_1, \ldots \bi{q}_N$ be an orthonormal basis in ${\cal S}$, and $\bi{q}_{N+1}, \ldots \bi{q}_M$ be an orthonormal basis in ${\cal S}^{\perp}$, the orthogonal complement of ${\cal S}$ in $\mathbf{C}^M$. If we set $W=[\bi{q}_1 \ldots \bi{q}_M]$ then (\ref{WX}) holds. Obviously, all other suitable unitary transformations are obtained from this $W$  by multiplying it to the right by the block diagonal unitary matrices $\diag[U,V]$ where $U$ and $V$ run through the unitary groups  $U(N)$ and $U(M-N)$, respectively. Multiplying $W$ by $\diag [\mathbbm{1}_N, V]$ corresponds to choosing a different orthonormal basis in  ${\cal S}^{\perp}$, and multiplying $W$ by $\diag [U, \mathbbm{1}_{M-N}]$ corresponds to replacing matrix $G$ by $UG$.

It is straightforward to check that any unitary matrix $W\in U(M)$  can be transformed to the block form
\be \label{bb11}
\fl W=\left[ \begin{array}{cc}(\mathbbm{1}_N-CC^{\dagger})^{1/2}& C \\-C^{\dagger} & (\mathbbm{1}_{M-N}-C^{\dagger}C)^{1/2}  \end{array} \right], \quad \text{where $C$ is $N\times (M-N)$} \; ,
\ee
by multiplying it to the right by block diagonal unitary matrix as above\footnote{
Representation (\ref{bb11}) can also be derived by exponentiating matrices of the form
${\cal A} =\Big[
\begin{array}{cc}
0 & A \\
-A^{\dagger} & 0
\end{array}
\Big] $
with $A$ being $N\times (M-N)$. Such matrices ${\cal A}$ form the orthogonal component of
$\mathtt{u}(N)\oplus \mathtt{u}(M-N)$ in $\mathtt{u}(M)$;
hence the matrices $W=\exp({\cal A})$ provide a natural set of representatives for
$U(M)/(U(N)\times U(M-N))$. Expanding the exponential in Taylor series leads to (\ref{bb11})  with $C=A(A^{\dagger}A)^{-1/2}\sin ((A^{\dagger}A)^{1/2})$.
}.
Correspondingly, we shall seek the unitary matrix $W$ in (\ref{WX}) in the block form (\ref{bb11}). This additional condition makes the decomposition in (\ref{WX}) unique for ('standing') rectangular matrices of full rank. Note that if the matrix $X$ is real then $G$, $W$, and, correspondingly, $C$ in (\ref{bb11}) can all be chosen real as well.

Having settled on the choice of $W$, we can solve the equation in (\ref{WX})  for $G$ and $C$.

\begin{lem} \label{lem1} Let $M > N$. Suppose that $Y$ is $N\times N$ and  $Z$ is $(M-N)\times N$, and $Y$ is invertible. Then there is a unique $M\times M$ unitary matrix $W$ of the form (\ref{bb11}) such that (\ref{WX}) holds. The square matrix $G$ in (\ref{WX}) is given by
\begin{equation}
G=
 \left(\mathbbm{1}_N+\frac{1}{Y^{\dagger}}Z^{\dagger}Z\frac{1}{Y} \right)^{1/2} Y \; .
\label{atilda}
\end{equation}
\end{lem}

\medskip

A direct calculation proving this lemma is provided in \ref{appendix_A}.

\medskip

We now have a procedure for quadratising 'standing' rectangular matrices. Of course, in the opposite case of  'lying'  rectangular matrices ($M<N$) one may apply the same procedure to quadratise the transposed matrix $X^T$. Thus,  any rectangular matrix $X$ can be quadratised by a unitary transformation on its columns (or rows, if the number of columns is greater than the number of rows), giving rise to a square matrix $G$. As $G$ is a unique solution of  equation (\ref{WX}), its spectrum characterises algebraic properties of the rectangular matrix $X$.

Motivated by studies of quantum operations, see Section \ref{sec:quant}, we want to explore the concept of quadratisation in the context of random rectangular matrices. To this end, we will  consider Gaussian random matrices $X$, real or complex, with $M$ rows and $N$ columns, $M> N$, so that the probability distribution of $X$ is specified by the measure
\be\label{bb0}
\diffi \nu (X) \propto \exp\left( -\frac{\beta}{2} \Tr X^{\dagger} X\right) |\diffi X|\; ,
\ee
where $\beta=1$ or $\beta=2$ depending on whether the matrix $X$ is real or complex, and $|\diffi X|$ is the flat (Lebesgue) measure on the corresponding matrix space.

Note that the result of quadratisation depends on the particular choice of elements of $W^{\dagger}X$ which are set to zero by action of the unitary $W$ in the ansatz (\ref{WX}). For instance, one could define another matrix ${G}'$ by assuming that unitary $W$ brings to zero the upper rectangular part of the matrix  $W^{\dagger} X$. However, we are going to apply (\ref{atilda}) for a random matrix $X$, drawn from a unitary invariant Gaussian ensemble. Thus the statistical properties of $G$ and ${G}'$ will be the same. Therefore for a rectangular random Gaussian matrix $X$ one may associate by (\ref{atilda})  a square random matrix
$G$, which we shall call the {\sl quadratisation} of $X$.

The probability distribution for square matrices $G$ (\ref{atilda})  induced by the normal distribution for $Y$ and $Z$ can be obtained directly from (\ref{WX}). However, it is useful to look at this problem from the Singular Value Decomposition (SVD) perspective, especially as one can utilise the known Jacobian of the corresponding coordinate transformation.

Ignoring a set of zero probability measure, the $N\times N$ matrix $X^{\dagger} X$ has $N$ distinct eigenvalues $s_j$, $0<s_1<s_2 < \ldots < s_N$, and the SVD asserts that $X$ can be factorised as follows
\be \label{bb1}
X = Q \; \Sigma^{1/2}  P^{\dagger},
\ee
where $\Sigma=\diag (s_1, \ldots , s_N)$, and $Q$ and $P$ are, respectively, $M\times N$ and $N\times N$ matrices with orthonormal columns, so that $Q^{\dagger}Q=P^{\dagger}P=\mathbbm{1}_N$. The columns of $P$ are in fact eigenvectors of $X^{\dagger}X$ and, hence, are  defined up to phase factor (or sign for real matrices).  To make the choice of $P$ unique, we shall impose the condition that the first non-zero entry in each column of $P$ is positive. Consequently,  the matrix $Q$ is also defined uniquely via $Q=XP\Sigma^{-1/2}$. Thus, factorisation (\ref{bb1}) introduces a new coordinate system  $(Q,P,\Sigma)$ in the space of rectangular matrices. In the new coordinates \cite{Forrester2010b}
\be \label{bb2}
\diffi \nu (X) = \diffi \mu (Q)\diffi \tilde \mu (P) \diffi \sigma (\Sigma)\; ,
\ee
where $\diffi \mu (Q)$ is the normalised invariant (Haar) measure on the (Stiefel) manifold of complex (or, correspondingly, real) $M\times N$ matrices with orthonormal columns,  $\diffi \tilde \mu (P)$  is the normalized measure defined by the maximum degree form
\[
\omega (P)= \left\{
\begin{array}{ll}
\wedge_{1\le k<j\le N} (\re (P^{\dagger}dP)_{jk} \wedge \im (P^{\dagger}dP)_{jk}) & \hbox{($\beta=2$),}\\
\wedge_{1\le k<j\le N} (P^{\dagger}dP)_{jk} & \hbox{($\beta=1$), }
\end{array}
\right.
\]
on the manifold of unitary (real orthogonal for $\beta=1$) $N\times N$ matrices satisfying the column condition above (the first non-zero entry in each column is positive), and
\be \label{bb3}
\diffi \sigma (\Sigma) \propto (\det \Sigma )^{\frac{\beta}{2}(M-N+1-\frac{2}{\beta})}\  e^{-\frac{\beta}{2}\Tr \Sigma}\  \prod_{j<k}  |s_k-s_j|^{\beta}\prod_{j=1}^N \diffi s_j\; .
\ee
It is apparent from (\ref{bb2}) that the matrices $Q, P$ and $\Sigma$ are mutually independent.

Let us introduce an additional unitary (real orthogonal for $\beta=1$) matrix $U$ of size $N\times N$ and rewrite (\ref{bb1}) in the form
\be \label{bb4}
X = QUU^{\dagger} \Sigma^{1/2} P^{\dagger} = QU G. 
\ee
The matrix $G=U^{\dagger} \Sigma^{1/2} P^{\dagger} $ is $N\times N$. Now choose $U$  to be Haar unitary (real orthogonal for $\beta=1$) and independent of $Q, P$ and $\Sigma$. Then,  by rolling back from (\ref{bb3}) to (\ref{bb1}) with $Q$ replaced by $U^{\dagger}$,  the square matrix $G$ is distributed according to the measure
\be \label{bb5}
\diffi \mu_{\mathrm{IndG}} (G) \propto (\det G^{\dagger}G)^{\frac{\beta}{2} (M-N)}\exp\left( -\frac{\beta}{2} \Tr G^{\dagger} G\right) |\diffi G|\; ,
\ee
with the determinant on the right in (\ref{bb5}) originating from the one in (\ref{bb3}). This is the induced Ginibre distribution introduced in Section \ref{sec:1}.   Because of the invariance of the Haar distribution,  the unitary matrix $U$ in (\ref{bb4}) can be absorbed into $Q$ . In other words,  we have decomposed the rectangular Gaussian $X$ into the product
\be \label{gs}
X=  \tilde Q G
\ee
of two independent random matrices: the rectangular matrix $\tilde Q:=QU$ with orthonormal columns, this one has uniform distribution, and a square matrix $G$, this one has induced Ginibre distribution.  Decomposition (\ref{gs})  can also be written as
\be \label{gs1}
X=W \left[\begin{array}{c}
G \\0 \end{array}\right]\; ,
\ee
where $W$ is an $M\times M$ unitary (real orthogonal for $\beta=1$) matrix obtained from the $M\times N$ matrix  $\tilde Q$ by appending suitable column-vectors. This is nothing else as equation (\ref{WX}). One can transform the matrix $W$ to the block form of  (\ref{bb11}), so that (\ref{gs1}) becomes
\[
X=\left[ \begin{array}{cc}(\mathbbm{1}_N-CC^{\dagger})^{1/2}& C \\-C^{\dagger} & (\mathbbm{1}_{M-N}-C^{\dagger}C)^{1/2}  \end{array} \right]  \left[\begin{array}{c}
\tilde G \\0 \end{array}\right],
\]
where $\tilde G= \tilde U G$ for some unitary $\tilde U$. Obviously, $\tilde G$ has the same distribution as $G$.  Thus:

\begin{lem}\label{lem2}
Suppose that $X$ is $M\times N$ Gaussian, $M>N$, with distribution (\ref{bb0}). Then its quadratisation $G$ has the induced Ginibre distribution (\ref{bb5}).
\end{lem}

This result holds in the cases of complex and real Ginibre matrices. It shows that the notion of quadratisation of a rectangular matrix is specially justified for Gaussian rectangular matrices.  In this case the statistical properties of the outcome do not depend  on the particular way, how the quadratisation is obtained out of the initially rectangular random matrix. This is in a clear analogy to the ensembles of truncations of  unitary \cite{ZS00} and orthogonal \cite{KSZ10} random matrices, the statistical properties of which do not depend on the choice of the rows and columns to be truncated. Thus for any random rectangular matrix $X$ we may introduce its {\sl quasi spectrum} as the spectrum of its quadratisation.

Lemma \ref{lem2} together with Lemma \ref{lem1} provides a recipe for generating matrices from the induced Ginibre distribution starting with Gaussian matrices. Interestingly, by rearranging  (\ref{bb4})  one obtains another recipe which is might be less efficient computationally but still interesting from the theoretical point of view.  Indeed,  since $(X^{\dagger}X)^{1/2}=P\Sigma^{1/2}P^{\dagger}$, it follows from (\ref{bb4}) that
\be \label{bb6}
G=U^{\dagger}Q^{\dagger} X=U^{\dagger}P^{\dagger}(X^{\dagger}X)^{1/2}=\tilde U^{\dagger}(X^{\dagger}X)^{1/2}\; , \quad \text{where $\tilde U=PU$}\; .
\ee
Recalling that the Haar measure is invariant with respect to right (and left) multiplication, one arrives at the following recipe for generating matrices from the induced Ginibre distribution.

\begin{lem}
Suppose that $U$ is $N\times N$ Haar unitary (real orthogonal for $\beta=1$) and $X$ is $M\times N$ Gaussian with distribution (\ref{bb0}) and independent of $U$. Then the $N\times N$ matrix $G=U(X^{\dagger}X)^{1/2}$  has the induced Ginibre distribution (\ref{bb5}).
\end{lem}

Obviously, our arguments extend to random rectangular matrices with invariant distributions other than Gaussian, e.g. the Feinberg-Zee distribution with density
\be \label{FZ1}
p_{\mathrm{FZ}} (X) \propto \exp [-\Tr V(X^{\dagger}X)  ] \;,
\ee
where $X$ is $M\times N$, $M>N$. On applying the procedure of quadratisation to such an ensemble, one obtains the induced Feinberg-Zee distribution
\[
p_{\mathrm{IndFZ}} (G) \propto  (\det G^{\dagger}G)^{\frac{\beta}{2} (M-N)} \exp [-\Tr V(G^{\dagger}G)  ]\; .
\]
Another example is provided by rectangular truncations of random unitary or orthogonal matrices\footnote{For a range of parameter values the matrix distribution of such truncations is given by (\ref{FZ1}) with  $V(X^{\dagger}X) $ replaced by a sum of powers of $\log X^{\dagger}X $ and $ \log(1- X^{\dagger}X) $, see \cite{Forrester2006}}. By applying quadratisation, one can extend the study of square truncations of Haar unitary \cite{ZS00} and orthogonal \cite{KSZ10} matrices.

In the context of eigenvalue maps, it is instructive to embed rectangular $M\times N$ matrices $X$ into the space of $M\times M$ matrices by augmenting $X$ with $M-N$ zero column-vectors and write the quadratisation rule (\ref{WX}) in terms of square matrices albeit with zero blocks:
\be \label{WX1}
W^{\dagger} \tilde X = \tilde G, \quad \text{with}\;\; \tilde X = \left[ \begin{array}{cc}Y & 0 \\Z & 0  \end{array} \right]\; \text{and}\; \tilde G= \left[ \begin{array}{cc}G & 0 \\0 & 0  \end{array} \right]\; ,
\ee
where as before $W$ is $M\times M$ unitary, $Y$ and $G$ are $N\times N$ and $Z$ is $(M-N)\times N$.

By construction, zero is an eigenvalue of $\tilde X$ of multiplicity $M-N$ and the remaining $N$ eigenvalues of $\tilde X$ are exactly those of its top left block $Y$. Thus our quadratisation procedure induces an eigenvalue map: the eigenvalues of $\tilde X$ are mapped onto the eigenvalues of $\tilde G$. Under this map, the zero eigenvalue stays put, and its multiplicity is conserved, and, otherwise, the eigenvalues of $Y$ are mapped onto those of $G$.

For Gaussian matrices the eigenvalues of $Y$, for large matrix dimensions ($M\gg 1$ and $N\propto M$) are distributed uniformly in a disk, and we shall show in the subsequent sections that the eigenvalues of $G$ are distributed uniformly in a ring. This hole in the spectrum created by quadratisation can be interpreted as due to repulsion of eigenvalues of $\tilde G$ from its zero eigenvalues.

\section{Complex induced Ginibre ensemble}
\label{sec:compl}
The complex induced Ginibre ensemble is defined by the probability density $p^{(2)}_{\mathrm{IndG}} (G)$ (\ref{ginibind}) -- (\ref{ginibind1}) on the set of  complex $N\times N$ matrices $G=[g_{jk}]$ with volume element $ |\diffi G|=\prod_{j,k=1}^N\diffi\Real g_{jk}\diffi\Image g_{jk}$. In the subsequent analysis we restrict ourselves to non-negative integer $L$, in which case the obtained results can be interpreted in the context of quadratisation of $M\times N$  matrices,  with $L=M-N$ being a measure of rectangularity as was discussed in Section \ref{w:section4.2}. However, our analysis extends almost verbatim to real non-negative $L$, see \cite{JF}.

The symmetrised joint probability density function $P(\lambda_1, \ldots, \lambda_N)$  of the eigenvalues in the complex induced Ginibre ensemble is obtained from that in the Ginibre ensemble \cite{Ginibre1965, Me04} by multiplying through by $\det(GG^\dag)^L=\prod_{j=1}|\lambda_j|^{2L}$ and re-evaluating the normalization constant. This yields
\begin{equation}\label{b9}
\fl P (\la_1,\ldots,\la_N)=\frac{1}{N!\pi^N}\prod_{j=1}^N\frac{1}{\Gamma(j+L)}
\prod_{j<k}|\la_k-\la_j|^2\prod_{j=1}^N|\la_j|^{2L}\exp\left(-\sum_{j=1}^N|\la_j|^2\right).
\end{equation}
Consequently, the $n$-eigenvalue correlation functions
\begin{equation*}
R_n(\la_1,\ldots,\la_n)=\frac{N!}{(N-n)!}\int P(\la_1,\ldots,\la_N)d^2\la_{n+1}\cdots d^2\la_N.
\end{equation*}
follow in the determinantal form,
\begin{equation}\label{b1a}
R_n(\la_1,\ldots,\la_n)=\det\Big(K_N(\la_k,\la_l)\Big)_{k,l=1}^n\; ,
\end{equation}
via the method of orthogonal polynomials in the same way as for the complex Ginibre ensemble \cite{Me04}. The kernel $K_N(\la_k,\la_l)$ is given by
\begin{equation}\label{b1}
K_N(\la_k,\la_l)=\frac{1}{\pi}e^{-\f|\la_k|^2-\f|\la_l|^2}
\sum_{j=0}^{N-1}\frac{\big(\la_k\bar{\la}_l\big)^{j+L}}{\Gamma(j+L+1)}\: .
\end{equation}
It should be noted that the right-hand side  in (\ref{b1}) can be expressed as a difference of two incomplete gamma functions. This makes the asymptotic analysis of the eigenvalue statistics for large matrix dimensions straightforward.

Of special interest is the mean eigenvalue density $\rho_N (\la):=R_1(\la)=K_N(\la,\la)$,
\be \label{b2}
\rho_N (\la) =  \frac{1}{\pi}e^{-|\la|^2}
\sum_{l=L}^{L+N-1}\frac{|\la|^{2l}}{l!}
=
\frac{1}{\pi}\Big[\frac{\gamma(L,|\la|^2)}{\Gamma(L)}-\frac{\gamma(L+N,|\la|^2)}{\Gamma(L+N)}\Big]\; ,
\ee
where $\gamma(a,z)$ is the lower incomplete gamma function \cite{AS},
\be \label{b3}
\gamma(a,z):=\int_0^z t^{a-1}e^{-t}dt\; .
\ee
In the limit when  $N$ is large and $L= N\alpha$, $\alpha >0$ (which corresponds to quadratisation of `standing' rectangular matrices of size $(N+L)\times N$) the eigenvalue distribution  (in the leading order) is uniform and supported by a ring about the origin with the inner and outer radii $r_{in}=\sqrt{L}$ and $r_{out}=\sqrt{L+N}$, respectively. More precisely,
\be \label{b4}
\lim_{N\to\infty}\rho_N (\sqrt{N}z)=\frac{1}{\pi}\left[\Theta(\sqrt{\al+1}-|z|) - \Theta(\sqrt{\al}-|z|)\right],
\ee
where $\Theta$ is the Heaviside function, $\Theta (x)=1$ if $x>0$, $\Theta (x)=0$ if $x<0$ and $\Theta (0)=\frac{1}{2}$.

Close to the circular edges of the eigenvalue support, for every angle $\phi$,
\be \label{b5}
\lim_{N\to\infty} \rho_N ((r_{out}+\xi)e^{i\phi})=\lim_{N\to\infty} \rho_N ((r_{in}-\xi)e^{i\phi})=\frac{1}{2\pi}\erfc(\sqrt{2} \xi )\; ,
\ee
where $\erfc(x)$ is the complementary error function,
\be \label{erfc1}
\erfc(x):=\frac{2}{\sqrt{\pi}}\int_x^\infty e^{-t^2}dt\; .
\ee
Hence the eigenvalue density falls from $1/\pi$ to zero at a Gaussian rate at the inner and outer boundaries of the eigenvalue support.  As was observed in \cite{Bo10} the scaling law (\ref{b5}) is universal, see Eqs. (42)--(43) in \cite{Bo10}.

The scaling limits of the eigenvalue correlation functions $R_n(\la_1,\ldots,\la_n)$ in the induced Ginibre ensemble can also be obtained from (\ref{b1a})--(\ref{b1}). Not surprisingly in the regime when $N\to\infty$ and $L\propto N$ one recovers the same expressions as in the Ginibre ensemble \cite{Ginibre1965,Forrester1999,BS}, both in the bulk and at the circular edges of the eigenvalue distribution, see \ref{appendix_d}.

%

A quick inspection of the joint probability density function (jpdf) of the eigenvalues (\ref{b9}) convinces that the determinantal power in front of the Gaussian weight in (\ref{ginibind}) repels eigenvalues from the origin. When the rectangularity index $L$ grows proportionally with $N$, the eigenvalues are actually displaced from the origin, resulting in the creation of a ring of eigenvalues, see Fig. \ref{exec}.

\begin{figure}[htbp]
\centering
\includegraphics[width=4.0in]{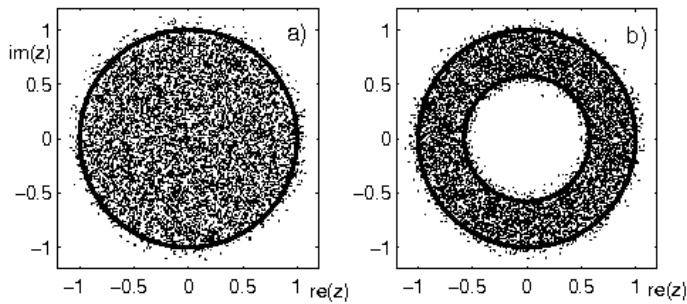} 
\caption[]{Spectra of matrices pertaining to
induced Ginibre ensemble of complex matrices for dimension $N=128$
and  a) $L=0$ and b) $L=32$.
Each plot consists of data from $128$ independent realisations.
The spectra are rescaled by a factor of ${1}/{\sqrt{N+L}}$ to be
localized inside the unit disk.
The circles of radius: $r_{\rm in}=\sqrt{L/(N+L)}$
(inner one) and $r_{\rm out}=1$ (outer one)
are depicted to guide the eye.
}
\label{exec}
\end{figure}

In the context of the augmented matrix $\tilde G$ (\ref{WX1}) the creation of the hole in the spectrum can be interpreted as due to repulsion from zero eigenvalues.  There are $M-N$ eigenvalues of $\tilde G$ located at the origin and the remaining eigenvalues are exactly as those of $G$. When $L$ increases, the zero eigenvalues become dominant and repel the rest of the eigenvalues strongly, which leads to the creation of the hole in the spectrum.  Looking at the polynomial on the right-hand side in (\ref{b2}),  it starts with the power $|\la|^{2L}$ and ends with the power $|\la|^{2(N+L-1)}=|\la|^{2(M-1)}$. Also in the $2-$dim plot in the complex plane we see that $(M-N)$ eigenvalues have been taken off from the centre near the origin and put in top of the surface of the circle of support. Thus the outer radius of the support of the density increases from $\sqrt{N}$ to $\sqrt{M}$ while there is an inner hole of radius $\sqrt{L}=\sqrt{M-N}$ - see Fig.~\ref{exec}.

Now we shall explore a different regime when the rectangularity index $L= M-N\ll N$. This corresponds to the quadratisation of almost square matrices. In the vicinity of the origin the corresponding large-$N$ (or equivalently large-$M$) limit can be performed by simply extending the summation in (\ref{b1}) to infinity. For the mean eigenvalue density, $\rho_N(\la)=R_1(\la)$ it gives\footnote{The complex induced Ginibre ensemble is a special case of the more general  non-Hermitian random matrix ensemble studied by Akemann in \cite{Akemann2001b} . Correspondingly, our (\ref{a1})--(\ref{a2})) is a limiting case of Eq. (4.5) in  \cite{Akemann2001b}. Taking the corresponding limit is straightforward for small values of $L$.}
\[
\lim_{N\to\infty} \rho_N(\la) =\frac{1}{\pi}e^{-|\la|^2}
\sum_{j=0}^{\infty}\frac{|\la|^{2(j+L)}}{\Gamma(j+L+1)}=\frac{1}{\pi} \frac{\gamma (L, |\la|^2)}{\Gamma (L)}\: .
\]
and more generally
\be \label{a1}
\lim_{N\to\infty}  R_n(\la_1, \ldots , \la_n)=\det \left(K_{origin}(\la_j, \la_k)\right)_{j,k=1}^n\; , \ee
with
\be \label{a2}
K_{origin}(\la_j, \la_k) = \frac{1}{\pi} \frac{\gamma (L, \la_j\bar \la_k)}{\Gamma (L)} =
\frac{1}{\pi} \ \frac{1}{\Gamma (L)}  \int_0^{\la_j\bar \la_k} t^{L-1}e^{-t}dt\; .
\ee
At the origin the eigenvalue density vanishes algebraically, $\rho_N (\la) \sim \frac{1}{\pi} \frac{|\la|^{2L}}{\Gamma (L+1)}$ as $\la \to\ 0$, uniformly in $N$. Away from the origin, the density reaches its asymptotic value $1/\pi$ very quickly\footnote{A similar behaviour is found in the chiral Ginibre ensemble, see \cite{Akemann2009} for a discussion in the context of gap probabilities.}.  This plateau extends to a full circle of radius $\sqrt{N}$:
\[
  \lim_{N\to\infty} \rho_N(\sqrt{N} z) =  \frac{1}{\pi}\,  \Theta (1-|z|)\; ,
\]
as in the Ginibre ensemble, and, moreover, for reference points $\sqrt{N}u$, $|u|<1$, one also recovers the Ginibre correlations. 

Another quantity of interest is the so-called hole probability $A(s)$ at the origin giving the probability that no eigenvalues lies inside the disk $D_s=\{z: |z|< s\}$. For finite $N$ the hole probability $A(s)$ can be derived from the expression:
\bes
A(s)=\int P(\la_1,\ldots,\la_N) \prod_{j=1}^N (1-\chi_{D_s}(\la_j)) d^2 \la_1 \ldots d^2 \la_N
\ees
where $\chi_{D_s}$ denotes the indicator function of $D_s$ by employing the method of orthogonal polynomials to yield
\be
A(s)=\prod_{j=1}^{N} \frac{\Gamma(j+L,s^2)}{\Gamma(j+L)}
\ee
with $\Gamma(a,x):=\int_x^\infty e^{-t}t^{a-1}dt$ denoting the upper incomplete gamma function. In the asymptotic regime of almost square matrices taking the large $N$ limit, while keeping $L$ fixed, results in the easily accessible expression for the hole probability $A(s)=1-\frac{s^{2(L+1)}}{(L+1)!}+O(\frac{s^{2(L+2)}}{(L+2)!})$.

\section{Real induced Ginibre ensemble}
\label{sec:real}

The real induced Ginibre ensemble is defined by the probability density $p^{(1)}_{\mathrm{IndG}} (G)$ (\ref{ginibind}) -- (\ref{ginibind1})  on the set of real $N\times N$ matrices $G=[g_{jk}]$ with volume element $|dG|=\prod_{j,k=1}^N\diffi g_{jk}$. In the following we restrict ourselves to even dimension $N$.

\subsection{The joint distribution of eigenvalues}

The difficulty in deriving the joint probability density function for real asymmetric matrices is due to the fact that there is a non-zero probability $p_{N,k}$ for the matrix $G$ to have $k$ real eigenvalues. In the following it is assumed that $G$ has $k$ real ordered eigenvalues: $\la_1\geq\ldots\geq\la_k$, while $l=\frac{N-k}{2}$ denotes the number of complex conjugate eigenvalue pairs $x_{1}\pm i y_{1},\ldots, x_{l}\pm i y_{l}$ ordered by their real part. We adopt the convention that $y_j>0$ for all $j$ .
In the case of two complex eigenvalues with identical real part the eigenvalue pairs are ordered by the imaginary part.

As a consequence the eigenvalue jpdf decomposes into a sum of probability densities $P_{N,k,l}(\la_1,\ldots,\la_k, x_{1} + i y_{1},\ldots,x_{l} + i y_{l})$,  corresponding to having $k$ real eigenvalues and $l$ pairs of complex conjugate eigenvalues. In order for $P_{N,k,l}$ to be non-zero $k$ must be even, as it is assumed that $N$ is even.

The derivation of the eigenvalue jpdf goes as follows \cite{E}, see \cite{LS,SW,KS} for alternative derivations. In order to change variables from the entries of $G$ to the eigenvalues of $G$ and some auxiliary variables the real Schur decomposition is employed: $G=QTQ^T$, where $Q\in \mathbb{R}^{N\times N}$ is an orthogonal matrix, whose first row is chosen to be non-negative and the matrix $T\in \mathbb{R}^{N\times N}$ is block triangular of the form:
\bes
T=\left(
    \begin{array}{cccccc}
      \la_1 & \cdots & t_{1k} & t_{1,k+1} & \cdots & t_{1,N} \\
       & \ddots & \vdots & \vdots &  & \vdots \\
      0 & & \la_k & t_{k,k+1} & \cdots & t_{k,N} \\
      0 & \cdots & 0 & Z_{1} & \cdots & t_{k+1,N} \\
      \vdots &  & \vdots &  & \ddots & \vdots \\
      0 & \cdots & 0 & 0&  & Z_l \\
    \end{array}
  \right)
  =\left(
            \begin{array}{cc}
              \Lambda & T^U \\
              0 & Z \\
            \end{array}
          \right)
  .
\ees
Here $\Lambda$ is triangular containing the real eigenvalues $\la_1,\ldots,\la_k$ of $G$ on its diagonal and $Z$ is block triangular containing the $2\times 2$ blocks:
\bes
      Z_j=\left(
      \begin{array}{cc}
        x_j & b_j \\
        -c_j & x_j \\
      \end{array}
    \right), \quad b_jc_j>0,\quad b_j\geq c_j\quad {\rm and}\quad y_j=\sqrt{b_jc_j}
\ees
on its block diagonal.

The Jacobian of this change of variable was already computed in \cite{E}: \bes |J|=2^l\Big|\Delta\big(\{\la_j\}_{j=1,\ldots,k}\cup \{x_j\pm i y_j\}_{j=1,\ldots,l}\big)\Big| \prod_{i>k}(b_i-c_i), \ees with $\Delta\big(\{z_p\}_{p=1,\ldots,n}\big):=\prod_{i<j}(z_j-z_i)$ denoting the Vandermonde determinant. Consequently we arrive at the relation:
\begin{eqnarray*}
p_{\mathrm{IndG}}^{(1)}(G)|dG|&=&C_{L}|J|\prod_{j=1}^k |\la_j|^L \prod_{m=1}^l (x_m^2+b_mc_m)^{L} \times \\ & &  e^{-\f\sum_{j=1}^k\la_j^2
-\f\sum_{t_{ij}} t_{ij}^2-\sum_{j=1}^l\big(x_j^2+\frac{b^2_j}{2}+\frac{c^2_j}{2}\big)}
|\diffi O||\diffi \Lambda||\diffi T^U||\diffi Z|.
\end{eqnarray*}
Here $|\diffi O|$ denotes the volume form on the space of orthogonal $N\times N$ matrices with positive first row while $|\diffi \Lambda|=\prod_{j=1}^k \diffi \la_j\prod_{m<n\leq k}^N\diffi t_{mn}$, $|\diffi T^U|=\prod_{m=1,\ldots,k, n=k+1,\dots, N}^N\diffi t_{mn}$ and $|\diffi Z|=\prod_{j=1}^l\diffi b_j \diffi c_j \diffi x_j\prod_{m,n= k+1}^N\diffi t_{mn}$. In addition another change of variable is necessary from the entries $x_j,b_j,c_j$ of the matrix blocks $Z_j$ to the real and imaginary part $x_j,y_j$ of the complex conjugate eigenvalue pairs and an auxiliary variable $\delta_j$. The change of variable is performed in the following way:
\begin{eqnarray}
\label{cov}
\fl{\rm Set}\qquad b_j=\f\left(\delta_j+\sqrt{\delta_j^2+4y_j^2}\right)\qquad{\rm and}\qquad c_j=\f\left(-\delta_j+\sqrt{\delta_j^2+4y_j^2}\right),
\end{eqnarray}
which implies $y_j=\sqrt{b_jc_j}$ and $\delta_j=b_j-c_j$. The Jacobian of this second change of variables can easily be determined:
\bes
|\bar{J}|=\frac{4y_j}{\sqrt{\delta_j^2+4y_j^2}}.
\ees
Integrating out the auxiliary variables $\delta_j$ for $j=1,\ldots,m$ and $t_{ij}$ as well as using $\Vol(O[N])=\frac{\pi^{\frac{1}{4}N(N+1)}}{\prod_{j=1}^N\Gamma(\frac{j}{2})}$ finally yields the partial eigenvalue joint probability density function:
\begin{eqnarray*}
\fl P_{N,k,l}(\la_1,\ldots,\la_k, x_{1} + i y_{1},\ldots,x_{l} +
i y_{l}) =
\frac{2^{2l-\frac{1}{4}N(N+1)}\pi^{-NL}}{\prod_{j=1}^{N}\Gamma(\frac{L+j}{2})}
 \Big|\Delta\big(\{\la_j\}_{j=1}^{k}\cup \{x_j\pm i
y_j\}_{j=1}^{l}\big)\Big| \times \\
\prod_{j=1}^k |\la_j|^L e^{-\f |\la_j|^2}
\prod_{m=1}^l (x_m^2+y_m^2)^{L} e^{y_m^2-x_m^2} y_m\erfc(\sqrt{2}y_m) \\
\end{eqnarray*}
where $\erfc(x)$ is the complementary error function (\ref{erfc1}), $\la_j\in\mathbb{R}$ for $j=1,\ldots,k$ and $x_m + i y_m\in \mathbb{C}_+$ for $m=1,\ldots,l$. Integrating the partial eigenvalue jpdf $P_{N,k,l}$ over $\mathbb{R}^k\times\mathbb{C}_+^{2l}$ gives $p_{N,k}$.

\subsection{The $(K',L')-$correlation functions}

Again we are interested in the correlations between the eigenvalues. The first starting point are the $(K',L',k',l')$-partial correlation functions which are just the symmetrised marginal probability density functions of $K'$ real eigenvalues and $L'$ complex eigenvalue pairs in the case that the number of real eigenvalues is $k'$ while the number of complex eigenvalues is $l'$ with different normalisation:
\begin{eqnarray*}
\fl R_{(K',L',k',l')}(\la_1,\ldots,\la_{K'},x_{1} + i y_{1},\ldots,x_{L'} + i y_{L'})
=\frac{k'!l'!2^{l'-L'}}{(k'-K')!(l'-L')!}\int_{\mathbb{R}^{k'-K'}}\int_{\mathbb{C}_+^{2(l'-L')}}\\
\fl P_{N,k',l'}(\la_1,\ldots,\la_{k'}, x_{1}+ i y_{1},\ldots,x_{l'}+ i y_{l'})d\la_{k'-K'+1}\cdots d\la_{k'} dx_{l'-L'+1}dy_{l'-L'+1}\cdots dx_{l'}dy_{l'}\; .
\end{eqnarray*}
The $(K',L')$-correlation functions which are the symmetrised marginals of $K'$ real eigenvalues and $L'$ complex eigenvalue pairs with different normalisation constant then decompose into a disjoint sum of probability density corresponding to having $k'$ real and $l'$ complex eigenvalues. They are defined as follows:
\begin{eqnarray}
\fl R_{K',L'}(\la_1,\ldots,\la_{K'},z_1,\ldots,z_{L'})
=\sum_{\substack{(K',L')\\K'\leq k',L'\leq l' }}R_{(K',L',k',l')}(\la_1,\ldots,\la_{K'},z_{1},\ldots,z_{L'})\; .
\end{eqnarray}
Remarkably it is possible to express the $(K',L')$-correlation functions in closed form using Pfaffians \cite{FN,BS,KS,SW}. An elegant approach to the derivation of the correlation functions is the method of skew-orthogonal polynomials \cite{FN,BS}. \\
\\
A family $\{q_j\}_{j=1,\ldots}$, of skew-orthogonal polynomials is said to be skew-orthogonal with respect to the skew-symmetric inner product $(-,-)$, if it satisfies
\begin{eqnarray}\label{skeworthogonalpolynomial}
(q_{2j},q_{2k})&=(q_{2j+1},q_{2k+1})=0\\
(q_{2j},q_{2k+1})&=-(q_{2j+1},q_{2k})=r_j\delta_{jk}\quad{\rm for}\quad {j,k=0,1,\ldots}\;\:.
\end{eqnarray}
In the context of the induced Ginibre ensemble $(-,-)$ denotes the skew-symmetric inner product:
\begin{eqnarray*}
(f,g)&:=(f,g)_{\mathbb{R}}+(f,g)_{\mathbb{C}}\\
(f,g)_{\mathbb{R}}&:=\int_{-\infty}^{\infty}\int_{-\infty}^{\infty}e^{-\f(x^2+y^2)}\sgn(y-x)|xy|^{L}f(x)g(y)dxdy\\
(f,g)_{\mathbb{C}}&:=2i\int_{\mathbb{R}_{+}^2}e^{y^2-x^2}\erfc(\sqrt{2}y)(x^2+y^2)^{L}\\
&\quad\left[f(x+iy)g(x-iy)-g(x+iy)f(x-iy)\right]dxdy.
\end{eqnarray*}
The method of skew-orthogonal polynomials leads to the following representation of the $(K'.L')$-correlation functions in terms of Pfaffians. A detailed derivation can be found in \cite{BS}. Denote
\begin{eqnarray*}
\tilde{q}(w)&:=e^{-\f w^2}w^L\sqrt{\erfc(\sqrt{2}\Image(w))} q(w)\\
\tau_j(w)&:=\cases{\f \int_{\mathbb{R}}\sgn(y-w)\tilde{q}_j(y)dy,&if $w\in\mathbb{R}$\\
i\tilde{q}_j(\bar{w})\sgn(\Image(w)),&if $w\in\mathbb{C}\backslash \mathbb{R}$}.
\end{eqnarray*}
Then
\be\label{Pfaffianentries}
\fl  R_{K',L'}(\la_1,\ldots,\la_{K'},z_1,\ldots,z_{L'})=\Pfaff\left[
\begin{array}{cc}
K_N(\la_j,\la_{j'}) & K_N(\la_j,z_{m'}) \\
K_N(z_m,\la_{j'}) & K_N(z_m,z_{m'})\\
\end{array}
\right].
\ee
with the $2\times 2$ matrix kernel
\be \label{Pfaffianentries1}
K_N(w,w'):=\left[\begin{array}{cc}
             DS_N(w,w') & S_N(w,w') \\
             -S_N(w,w') & IS_N(w,w')+\varepsilon(w,w')
           \end{array}\right],
\ee
where
\begin{eqnarray*}
DS_N(w,w')&=2\sum_{j=0}^{\frac{N}{2}-1}\frac{1}{r_j}\left[\tilde{q}_{2j}(w)\tilde{q}_{2j+1}(w')-\tilde{q}_{2j+1}(w)\tilde{q}_{2j}(w')\right]\\
S_N(w,w')&=2\sum_{j=0}^{\frac{N}{2}-1}\frac{1}{r_j}\left[\tilde{q}_{2j}(w)\tau_{2j+1}(w')-\tilde{q}_{2j+1}\tau_{2j}(w')\right]\\
IS_N(w,w')&=2\sum_{j=0}^{\frac{N}{2}-1}\frac{1}{r_j}\left[\tau_{2j}(w)\tau_{2j+1}(w')-\tau_{2j+1}(w)\tau_{2j}(w')\right]\\
\varepsilon(w,w')&=\cases{\f\sgn(w-w'),&if $\quad w,w'\in\mathbb{R}$\\
0,&else.}
\end{eqnarray*}
The indices  $j$ and $j'$ in (\ref{Pfaffianentries}) run from $1$ to $K'$ whilst $m$ and $m'$ run from $1$ to $L'$, so that the matrix inside the Pfaffian has the block structure with the top left and right bottom blocks being of size $2K'\times 2K'$ and $2L'\times 2L'$, respectively.

The entries of the Pfaffian kernel depend on the family of polynomials $q_j$ which are skew-orthogonal with respect to the inner product $(-,-)$. The direct computation of these polynomials is a tremendous task. As a result a different approach is employed in order to determine the required skew-orthogonal polynomials and thus the kernel entries in (\ref{Pfaffianentries1}).

As already observed in \cite{Sommers2007,Akemann2009} for the real Ginibre ensemble the following relationship for the entry $DS$ of the Pfaffian kernel in (\ref{Pfaffianentries1}) holds true:
\begin{eqnarray}\nonumber
\fl 2\sum_{j=0}^{\frac{N}{2}-1}\frac{1}{r_j}\left[q_{2j}(w)q_{2j+1}(w')-q_{2j+1}(w)q_{2j}(w')\right]
=\frac{1}{r_N}\times \\ (w-w')\langle\det\left(G-wI\right)\det\left(G-w'I\right)\rangle_{N-2},
\label{bb31}
\end{eqnarray}
where $\langle\ldots \rangle_{N-2}$ denotes the average over the  induced Ginibre  ensemble of square matrices $G$ of size $N-2$,  and  $r_N$ is the normalisation of the $N-$th skew-orthogonal polynomial as defined in (\ref{skeworthogonalpolynomial}).

On integrating out the `angular' part of $G$ in (\ref{bb31}), one is left with the integral over the `radial' part of $G$, see e.g. \cite{WK},
\begin{eqnarray*}
\langle\det\left(G-wI\right)\det\left(G-w'I\right)\rangle_{N-2}=\sum_{j=0}^{N-2}\frac{\langle\epsilon_j(GG^T)\rangle_{N-2}}{{n\choose j}}(ww')^{N-2-j}
\end{eqnarray*}
Here $\epsilon_j(GG^T)$ denotes the $j-$th elementary symmetric polynomial in the eigenvalues of $GG^T$. The above relation can also be obtained by expanding the product of determinants on the left hand side in Schur polynomials in the eigenvalues of $G$ and then averaging over the orthogonal group (which is the angular part of $G$), see \cite{FK2007} for a similar integral over the unitary group.

The average $\langle\epsilon_j(GG^T)\rangle_{N-2}$ can be reduced to a Selberg-Aomoto integral  \cite{Me04}. As a result
\begin{eqnarray}
\label{kerneleq}
\fl 2\sum_{j=0}^{\frac{N}{2}-1}\frac{1}{r_j}\left[q_{2j}(w)q_{2j+1}(w')-q_{2j+1}(w)q_{2j}(w')\right]
=\frac{1}{r_N} \times\nonumber\\
(w-w')\frac{\Gamma(L+N-1)}{\sqrt{2\pi}}\sum_{j=0}^{N-2}\frac{(ww')^j}{\Gamma(L+j+1)}\; .
\end{eqnarray}
As already observed in \cite{A} the skew-orthogonal polynomials can now be just "read off" using the fact that each $q_j$ is monic and of degree $j$ by for example differentiating:
\begin{eqnarray*}
\fl q_{2j}(w)=\left. r_{j}\frac{1}{(2j+1)!}\frac{\partial^{2j+1}}{\partial u^{2j+1}}\left[\frac{1}{r_N}(u-w)\Gamma(L+N-1)\sum_{j=0}^{2j}\frac{(wu)^j}{\Gamma(L+j+1)}\right]\right|_{u=0}\\
\fl q_{2j+1}(w)=r_{j}\frac{1}{(2j)!}\frac{\partial^{2j}}{\partial u^{2j}}\left.\left[\frac{1}{r_N}(w-u)\Gamma(L+N-1)\sum_{j=0}^{2j}\frac{(wu)^j}{\Gamma(L+j+1)}\right]\right|_{u=0}
\end{eqnarray*}
Hence for $j=1,2,\ldots$ the following polynomials were found to be skew-orthogonal with respect to the skew-inner product $(-,-)$:
\bes
q_{2j}(w)=w^{2j}\qquad q_{2j+1}(w)=w^{2j+1}-(2j+L)w^{2j-1}
\ees
In addition to that the first two skew-orthogonal polynomials are given by: $q_0(w)=1$ and $q_1(w)=w$. Similarly the normalisation constant can be found by comparison:
\bes
(q_{2j},q_{2j+1})=2\sqrt{2\pi}\Gamma(L+2j+1).
\ees
Thus the entries of the Pfaffian kernel can be explicitly determined. A detailed derivation of the necessary computation can be found in \ref{appendix3}. Here we just state the final result in Theorem~\ref{Pfaffian1} below.

Let
\bea
\nonumber
\psi(z)&=e^{-\f z^2}\sqrt{\erfc(\sqrt{2}\Image(z))}\\
t(x,z)&=\frac{1}{\sqrt{2\pi}}\psi(z)2^{\frac{L}{2}-1}z^{L}\frac{\Gamma(\frac{L}{2},\f x^2)}{\Gamma(L)}\; ,
\label{t_n}
\eea
where $\Gamma (a,z)$ is the upper incomplete Gamma function,
\[
\Gamma (a,z) = \int_z^{\infty} t^{a-1} e{-t}\ dt\; ,
\]
and
\begin{eqnarray}\nonumber
s_N(z,w) &=\frac{1}{\sqrt{2\pi}}\psi(z)\psi(w)\sum_{j=0}^{N-2}\frac{\big(wz)^{j+L}}{\Gamma(L+j+1)}\\
\label{r_n}
r_N(x,z)& =\frac{1}{\sqrt{2\pi}}\psi(z)\sgn(x)2^{\frac{N}{2}+\frac{L}{2}-\frac{3}{2}}z^{N+L-1}\frac{\gamma(\frac{N}{2}+\frac{L}{2}-\f,\f x^2)}{\Gamma(N+L-1)}\\
\end{eqnarray}
Note that $s_N(z,w)$ is symmetric in its variables and $t(x,z)$ and $r_N(x,z)$  are not.

\begin{thm}\label{Pfaffian1} The entries of the complex/complex ($2\times 2$) matrix kernel $K_N(z,w)$  in (\ref{Pfaffianentries})--(\ref{Pfaffianentries1}) are given by:
\begin{eqnarray*}
\fl DS_N(z,w)=(w-z)s_N(z,w); \; S_N(z,w)=i(\bar{w}-z)s_N(z,\bar{w}); \;  IS_N(z,w)=(\bar{z}-\bar{w})s_N(\bar{z},\bar{w}).
\end{eqnarray*}
The entries of the real/complex and complex/real matrix kernels $K_N(x,z)$ and $K_N(z,x)$ in (\ref{Pfaffianentries})--(\ref{Pfaffianentries1}) are given by:
\begin{eqnarray*}
\fl DS_N(x,z)=(z-x)s_N(x,z);\quad & DS_N(z,x)=-DS_N(x,z); \\
\fl S_N(x,z)=i(\bar{z}-x)s_N(x,\bar{z}); \quad &S_N(z,x)=s_N(x,z)+r_N(x,z)+t(x,z);\\
\fl IS_N(x,z)=-is_N(x,\bar{z})-ir_N(x,\bar{z})-it(x,\bar{z}); \quad & IS_N(z,x)=-IS_N(x,z) \; .
\end{eqnarray*}
And finally, the entries of  the real/real matrix kernel $K_N(x,y)$  in (\ref{Pfaffianentries})--(\ref{Pfaffianentries1}) are given by:
\begin{eqnarray*}
\fl DS_N(x,y)=(y-x)s_N(x,y); \quad S_N(x,y)=s_N(x,y)+r_N(y,x)+t(x,y);\\
\fl IS_N(x,y)=\frac{1}{\sqrt{2\pi}}\Big[-\frac{\gamma(L,y^2)}{\Gamma(L)}+e^{-\f(x-y)^2}\frac{\gamma(L,xy)}{\Gamma(L)}+\frac{y^Le^{\f y^2}}{\Gamma(L)}\int_x^ye^{-\f t^2}t^{L-1}dt\\
\fl+\frac{\gamma(L+N-1,y^2)}{L+N-1}-e^{-\f(x-y)^2}\frac{\gamma(L+N-1,xy)}{\Gamma(L+N-1)}-\frac{y^{L+N-1}e^{\f y^2}}{\Gamma(L+N-1)}\int_x^ye^{-\f t^2}t^{L+N-2}dt\\
\fl-\sgn(y)2^{\frac{L}{2}+N-\frac{3}{2}}\frac{\gamma(\frac{L}{2}+\frac{N}{2}-\f,\f y^2)}{\Gamma(L+N-1)}\int_x^ye^{-\f t^2}t^{L+N-1}dt
-2^{\frac{L}{2}-1}\frac{\Gamma(\frac{L}{2},\f y^2)}{\Gamma(L)}\int_x^ye^{-\f t^2}t^{L}dt\Big]
\end{eqnarray*}
\end{thm}

\subsection{Real and complex eigenvalue densities}

The eigenvalue densities for finite matrix dimensions $N$ can be read from the $(K',L')$-correlation functions (\ref{Pfaffianentries1}) specialising to the $(0,1)$ and $(1,0)$ cases. Indeed,
\be \label{rho_N_C}
\rho_N^C(z):=R_{0,1}(-,z)=\Pfaff K_N(z,z)=S_N(z,z)\quad (z\in \mathbb{C}_{+})\; ,
\ee
is the mean density of complex eigenvalues whilst
\be \label{rho_N_R}
\rho_N^R(x):=R_{1,0}(x,-)=\Pfaff K_N(x,x)=S_N(x,x)\quad (x\in \mathbb{R})
\ee
is the mean density of real eigenvalues.  Note the normalisation
\[
2\int_{\mathbb{C}_{+}}\rho_N^C(z)\ d^2 z + \int_{\mathbb{R}}\rho_N^R(x)\  d x =N\; .
\]
Theorem~\ref{Pfaffian1} now yields the finite-$N$ complex and real eigenvalue densities in a closed form
\begin{eqnarray}
\fl \rho_N^C(x+iy)&=\sqrt{\frac{2}{\pi}}y\erfc(\sqrt{2}y)e^{y^2-x^2}\sum_{j=0}^{N-2}\frac{(x^2+y^2)^{j+L}}{\Gamma(j+L+1)} \label{bb53} \\
\fl&=\sqrt{\frac{2}{\pi}}y\erfc(\sqrt{2}y)e^{2y^2}
\left[\frac{\gamma(L,x^2+y^2)}{\Gamma(L)}-\frac{\gamma(L+N-1,x^2+y^2)}{\Gamma(L+N-1)}\right] ,
\label{R1rec}
\end{eqnarray}
and
\begin{eqnarray}
 \fl \rho_N^R(x)&=\frac{1}{\sqrt{2\pi}}e^{-x^2}\sum_{j=0}^{N-2}\frac{x^{2(j+L)}}{\Gamma(j+L+1)}
+t(x,x)+r_N(x,x) \label{bb54} \\
\fl&=\frac{1}{\sqrt{2\pi}}\Big[\frac{\gamma(L,x^2)}{\Gamma(L)}-\frac{\gamma(L+N-1,x^2)}{\Gamma(N+L-1)}\Big]
+t(x,x)+r_N(x,x)\; ,
\label{R1rer}
\end{eqnarray}
where $t(x,x)$ and  $r_n(x,x)$ are the functions defined in (\ref{t_n}) -- (\ref{r_n}).


\begin{figure}[htbp]
\centering
\includegraphics[width=4.0in]{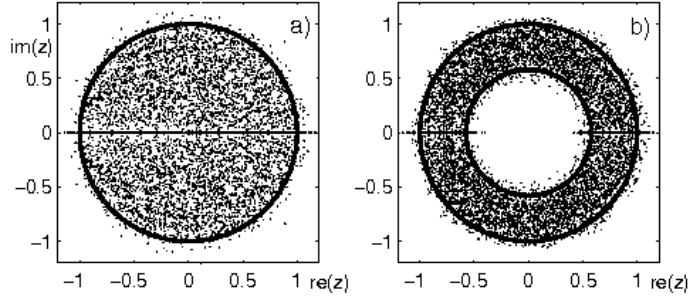} 
\caption[]{Spectra of random matrices from real induced Ginibre ensemble
for $N=128$
 and a) $M=N=128$ (no hole)  and b) $M=N+32$.
 Each picture consists of $128$ independent realisations and the spectra
are rescaled by a factor $1/\sqrt{M}$
as in Fig.~\ref{exec}.
In analogy to {\sf GinOE} we recognise a non-trivial
fraction of real eigenvalues.
}
\label{exer}
\end{figure}

\subsection{Asymptotic Analysis}

In this section we shall investigate the real induced Ginibre ensemble in the scaling limit when the free parameter $L$ grows proportionally with the matrix dimension $N$, which, in the language of quadratisation of rectangular matrices, corresponds to tall rectangular matrices which are neither skinny nor  almost-square.

The real and complex eigenvalue densities $\rho_N^C(z)$  (\ref{R1rec})  and  $\rho_N^R(x) $ (\ref{R1rer}) are already in a convenient form for the asymptotic analysis.  The well known limit relation for the error function \cite{AS},
\bes
\lim_{N\to\infty}\sqrt{N}we^{2Nw^2}\erfc(\sqrt{2N}w)=\frac{1}{\sqrt{2\pi}}\; ,
\ees
combined with the saddle-point analysis of the integral in (\ref{b3}) quickly gives the limiting (mean) density of complex eigenvalues. In the leading order, the distribution of complex eigenvalues turns out to be uniform in an annulus with the inner and outer radii $r_{in}=\sqrt{L}$ and $r_{out}=\sqrt{L+N}$, exactly as in the complex induced Ginibre ensemble.

Similarly, the saddle-point analysis of each of the incomplete Gamma functions in  (\ref{R1rer}) yields the limiting density of real eigenvalues.  In the leading order, the real eigenvalues in the induced Ginibre ensemble populate two symmetric segments  of the real axis, $[r_{in}, r_{out}]$  and $[-r_{out}, -r_{in}]$, with constant density. The theorem below summarises our findings and Figs.~\ref{exer}~and~\ref{cx_den} provide a comparison of the analytic results versus numerical simulations.

\begin{thm} \label{bb:thm2}
Suppose that $L=N\al$ with $\al >0$. Then:
\begin{enumerate}[(a)]
\item
In the leading order as $N\to\infty$, the average number of real eigenvalues in the real induced Ginibre ensemble is $\sqrt{\frac{2}{\pi}}(\sqrt{L+N}-\sqrt{L})$ and the density of  real eigenvalues
obeys the following limiting relation:
\begin{eqnarray*}
\lim_{N\to\infty}\rho_N^R(\sqrt{N}x)=\frac{1}{\sqrt{2\pi}}\left[\Theta(|x|-\sqrt{\al})-\Theta(|x|-\sqrt{\al+1})\right]\; .
\end{eqnarray*}
\item
The density of complex eigenvalues obeys the following limiting relation:
\begin{eqnarray*}
\lim_{N\to\infty}\rho_N^C(\sqrt{N}z)=\frac{1}{\pi}\left[\Theta(|z|-\sqrt{\al})-\Theta(|z|-\sqrt{\al+1})\right].
\end{eqnarray*}
\end{enumerate}
\end{thm}
On setting $L=0$ in the above results one recovers  the expected number $\sqrt{2N/\pi}$ of real eigenvalues in the Ginibre ensemble \cite{Edelman1994} together with the uniform densities of distribution of real and complex eigenvalues \cite{Edelman1994, E}.

One can examine how quickly the eigenvalue density falls to zero when one moves away from the boundary of the eigenvalue support. At the inner and outer circular edges \emph{away from the real line},  one recovers the same eigenvalue density profile as for the complex Ginibre ensemble \cite{Forrester1999}.

\begin{thm}\label{bb:thm3}
Suppose that $L=N\al$ with $\al >0$. Then, for fixed  $\xi \in \mathbb{R}$ and $ \phi \not=0, \pi$:
\be \label{bb50}
\lim_{N\to\infty} \rho_N^C( (\sqrt{L}-\xi)e^{i\phi})= \lim_{N\to\infty} \rho_N^C ( (\sqrt{L+N}+\xi)e^{i\phi})=\frac{1}{2\pi}\erfc(\sqrt{2}\xi).
\ee
\end{thm}
It follows from (\ref{bb50}) that the density of complex eigenvalues in the real induced Ginibre ensemble falls to zero very fast (at a Gaussian rate) away from the boundary of the eigenvalue support. This is also true of the density of real eigenvalues:

\begin{thm}\label{bb:thm4}
Suppose that $L=N\al$ with $\al >0$. Then, for fixed  $\xi \in \mathbb{R}$,
\bes
\fl \lim_{N\to\infty}\rho_N^R(\sqrt{L}-\xi)=\lim_{N\to\infty}\rho_N^R (\sqrt{L+N}+\xi)= \frac{1}{\sqrt{2\pi}}\left[\erfc(\sqrt{2}\xi)+\frac{1}{2\sqrt{2}}e^{-\xi^2}\erfc(-\xi)\right].
\ees
\end{thm}

\medskip

The proofs of Theorems~\ref{bb:thm2},~\ref{bb:thm3}~and~\ref{bb:thm4} are straightforward but tedious and are omitted.\\

Another interesting transitional region appears close to the real line. Here the density of complex eigenvalues is more sparse: for finite matrix dimensions $\rho_N^C(x+iy)\propto y$ for small values of $y$. One can easily obtain the complex eigenvalue density profile in the crossover from zero density on the real axis to the plateau  of constant density far away from the real axis. For example, at the origin
$
\lim_{N\to\infty}\rho_N^C(iv)=\sqrt{\frac{2}{\pi}}v\erfc(\sqrt{2}v)e^{2v^2}
$,
and more generally
\bes
\lim_{N\to\infty}\rho_N^C (\sqrt{N} u + i v)=\sqrt{\frac{2}{\pi}}v\erfc(\sqrt{2}v)e^{2v^2} \quad \text{for $|u| \in (\sqrt{\al}, \sqrt{\al +1})$}.
\ees

It should be noted that apart from the support of the eigenvalue distribution which clearly depends on $\al$, the limiting eigenvalue density profiles in various scaling regimes in the induced Ginibre ensemble are independent of $\al $ and coincide with those for the original Ginibre ensemble. This correspondence also extends to the eigenvalue correlation functions. We show in \ref{appen_C} that the eigenvalue correlation functions in the induced Ginibre ensemble in the bulk and at the edges are given by the expressions obtained for the Ginibre ensemble \cite{BS}, see also \cite{FN, Sommers2007, Forrester2008}.

\begin{figure}[]
\centering
\includegraphics[width=5.0in]{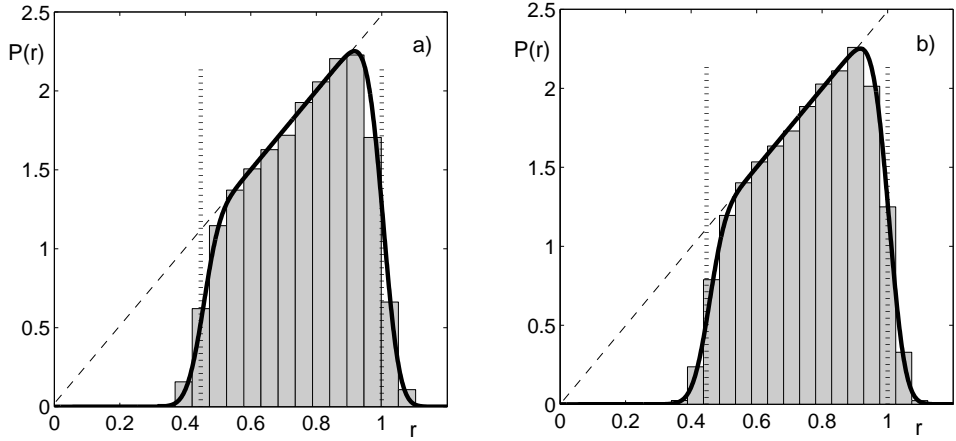}
\caption[]{
Radial density of complex eigenvalues of $256$ realisations of
a) complex and b) real
induced Ginibre matrices of dimension $N=128$ with $L=32$ compared
with the analytical results (solid lines)
obtained from Eq.~(\ref{b2}) and~(\ref{R1rec}), respectively.
Linear character of the curves between the
inner radius, $r_{\rm in}=\sqrt{L}/\sqrt{L+N}\approx 0.447$ 
and the outer radius, $r_{\rm out}=1$,
(both rescaled by $1/\sqrt{L+N}$ and represented by vertical lines)
suggests that the
distribution of eigenvalues on the ring is uniform in both cases.
Since the dimension $N$ is relatively small, the area on which
$P(z)$ is linear does not cover the entire interval $[r_{\rm in},r_{\rm out}]$.
}
\label{cx_den}
\end{figure}

\subsection{Almost square matrices}
Another interesting regime arises when the rectangularity index $L$ is fixed instead of growing proportionally with matrix size as discussed toward the end of Section \ref{sec:compl}. In the bulk, i.e. at a distance of order $\sqrt{N}$ from the origin one recovers uniform distribution of eigenvalues (real and complex) and Ginibre correlations, whereas in the vicinity of the origin new eigenvalue statistics arise. The eigenvalue densities can be obtained by extending the summation in (\ref{bb53}), (\ref{bb54}) to infinity. This yields
\begin{eqnarray*}
\lim_{N\to\infty}\rho_N^C(x+iy)
=\sqrt{\frac{2}{\pi}}y\erfc(\sqrt{2}y)\, e^{2y^2}\, \frac{\gamma(L,x^2+y^2)}{\Gamma(L)}
\end{eqnarray*}
for the density of complex eigenvalues
and
\begin{eqnarray*}
\lim_{N\to\infty}\rho_N^R(x)&=\frac{1}{\sqrt{2\pi}}\Big[\frac{\gamma(L,x^2)}{\Gamma(L)}
+e^{-\f x^2}x^{L}2^{\frac{L}{2}-1}\frac{\Gamma\Big(\frac{L}{2},\f x^2\Big)}{\Gamma(L)}\Big]
\end{eqnarray*}
for the density of real eigenvalues.

As in the case of complex matrices the higher order correlation functions at the origin are non-universal:
\begin{enumerate}
  \item The limiting real/real kernel is given by a $2\times 2$ matrix
    \bes
   \fl K_{origin}(r,r')=\frac{1}{\sqrt{2\pi}}\left[
              \begin{array}{cc}
                (r'-r)e^{-\f(r-r')^2}\frac{\gamma(L,rr')}{\Gamma(L)} & e^{-\f(r-r')^2}\frac{\gamma(L,rr')}{\Gamma(L)}+t(r,r') \\
                -e^{-\f(r-r')^2}\frac{\gamma(L,rr')}{\Gamma(L)}-t(r,r') & (*) \\
              \end{array}
            \right]
    \ees
    where
    \bes
    \fl(*)=-\frac{\gamma(L,r'^2)}{\Gamma(L)}+e^{-\f(r-r')^2}\frac{\gamma(L,rr')}{\Gamma(L)}+\Big(\frac{r'^Le^{\f r'^2}}{\Gamma(L)}-2^{\frac{L}{2}-1}\frac{\Gamma\Big(\frac{L}{2},\f r'^2\Big)}{\Gamma(\frac{L}{2})}\Big)\int_x^y e^{\f t}t^Ldt
    \ees
  \item The limiting complex/complex kernel is given by a $2\times 2$ matrix
\begin{eqnarray*}
    \fl K_{origin}(z,z')=\frac{1}{\sqrt{2\pi}}\sqrt{\erfc(\sqrt{2}\Image(z))\erfc(\sqrt{2}\Image(z'))}\\
\times\left[              \begin{array}{cc}
               (z-z')e^{-\f(z-z')^2}\frac{\gamma(L,zz')}{\Gamma(L)} & i(\bar{z}-z')e^{-\f(z-\bar{z'})^2} \frac{\gamma(L,z\bar{z'})}{\Gamma(L)} \\
                i(z'-\bar{z})e^{-\f(\bar{z}-z')^2}\frac{\gamma(L,z\bar{z'})}{\Gamma(L)}  & (\bar{z}-\bar{z'})e^{-\f(\bar{z}-\bar{z'})^2}\frac{\gamma(L,\bar{z}\bar{z'})}{\Gamma(L)} \\
              \end{array}
            \right]
    \end{eqnarray*}
    \\

  \item The limiting real/complex kernel is given by a $2\times 2$ matrix
\begin{eqnarray*}
    \fl K_{origin}(r,z)=\frac{1}{\sqrt{2\pi}}\sqrt{\erfc(\sqrt{2}\Image(z))}\\
    \times\left[\begin{array}{cc}
                (z-r)e^{-\f(r-z)^2}\frac{\gamma(L,rz)}{\Gamma(L)} & i(\bar{z}-r)e^{-\f(r-\bar{z})^2}\frac{\gamma(L,r\bar{z})}{\Gamma(L)} \\
                -e^{-\f(r-z)^2}\frac{\gamma(L,r\bar{z})}{\Gamma(L)} & -i\e^{-\f(r-\bar{z})^2}\frac{\gamma(L,r\bar{z})}{\Gamma(L)}-it(r,\bar{z}) \\
              \end{array}
            \right]
    \end{eqnarray*}
\end{enumerate}
Nevertheless setting the reference points at a distance of $\sqrt{N}$ away from the origin then yields the universal Ginibre correlation functions as given in \ref{appen_C}.

\section{Quantum operations and spectra of associated evolution operators}
\label{sec:quant}

We start this section with a brief review of quantum maps which act
on the set of density operators. Later on we identify a class of maps,
for which the associated one--step evolution operator is
represented by a rectangular  real matrix.
We argue that for a generic random map from this class
such an operator may be described by the
induced Ginibre ensemble of real matrices.

\subsection{Quantum maps}

A quantum state acting on a $d$--dimensional Hilbert space ${\cal H}_d$
can be represented by a density matrix $\rho$ of size $d$.
It is a positive  Hermitian matrix, $\rho^{\dagger}=\rho\ge 0$,
normalised  by the trace condition, ${\rm tr}\rho=1$.
Let ${\cal M}_d$  denote the set of all density matrices of size $d$.

Consider a linear quantum map $\Phi$, which maps the set
of density matrices onto itself,  $\Phi: {\cal M}_d \to {\cal M}_d$.
Such a map is called {\sl positive} as it transforms a positive operator
into a positive operator.  In quantum theory one uses also
a stronger property: a map is called {\sl completely positive}
if any extended map,   $\Phi \otimes  \mathbbm{1}_n$
is positive for an arbitrary dimension $n$ of the auxiliary subsystem.

Any completely positive map $\Phi$ can be written in the following Kraus form,
\begin{equation}
\rho'=\Phi(\rho) = \sum_{i=1}^k  E_i \rho E_i^\dagger .
\label{kraus1}
\end{equation}
The Kraus operators $\{ E_i\}_{i=1}^k$ and their number $k$ can be arbitrary
- see e.g. \cite{BZ06}.
The map  (\ref{kraus1}) preserves the trace, ${\rm tr} \rho={\rm tr} \rho'$,
if the entire set of $k$ Kraus operators satisfies the identity resolution,
\begin{equation}
\sum_{i=1}^k  E_i^\dagger E_i = \mathbbm{1} .
\label{trace2}
\end{equation}
Any physical transformation of a quantum state $\rho$
can be described by a completely positive, trace preserving map,
which is  also called  a {\sl quantum operation}.

Alternatively, any quantum operation can be
written as a result of a unitary operation acting on an extended system
and followed by the partial trace over the auxiliary subsystem $E$,
\begin{equation}
\rho'= \Phi(\rho)={\rm tr_{E}}\left[U(\rho\otimes |\nu\rangle \langle \nu|)U^{\dagger}\right].
\label{env_rep}
\end{equation}
Here a unitary operator $U\in U(kd)$ acts on a composite Hilbert space
${\cal H}_S \otimes {\cal H}_E$, where $S$ denotes the principal system of size $d$,
while $E$ denotes an environment of size $k$.
It is assumed that the environment is described initially by a pure state,
 $|\nu\rangle \in {\cal H}_E$.

The action of the linear transformation (\ref{kraus1})
can be described by a matrix $\Phi$ of size $d^2$,
\begin{equation}
\rho^\prime \; = \;  \Phi \rho \quad \quad {\rm or} \quad \quad
\rho_{m\mu}\!\!\!\!\!\!^\prime\ \ \: = \:
 \Phi_{\stackrel{\scriptstyle m \mu}{n \nu}}
  \, \rho_{n \nu} \ ,
 \label{map2}
\end{equation}
where the summation over repeated indices goes from $1$ to $d$.
Since the matrices $\rho$ and $\rho'$ represent
quantum states, so they are positive and normalised,
the following matrix with exchanged indices,
$ D_{\stackrel{\scriptstyle m n}{\mu \nu}}
: = \Phi_{\stackrel{\scriptstyle m \mu}{n \nu}}$
is also Hermitian, positive and normalised, tr$D=d$.
Thus this matrix of rank $k$, called
{\sl dynamical matrix}
is related to a quantum state $\varsigma$ acting on the
composed Hilbert space ${\cal H}_d \otimes {\cal H}_d$.
The state $\varsigma$, proportional to the dynamical matrix,
 can be expressed  by an extended map \cite{BZ06}
\begin{equation}
\varsigma=\frac{1}{d} D \; = \;
(\Phi \otimes {\mathbbm 1})|\psi^+\rangle \langle \psi^+| \ ,
 \label{dyn}
\end{equation}
where $|\psi^+\rangle$
denotes the maximally entangled state,
$|\psi^+\rangle=\frac{1}{d}\sum_{i=1}^d |i\rangle \otimes |i\rangle$.


The dynamical matrix $D$ is Hermitian by construction,
 while the superoperator matrix $\Phi$ is in general not Hermitian.
Positivity of the map and the trace preserving condition
imply, that the spectrum of $\Phi$ is contained in the unit disk
and there exists at least one eigenvalue equal to unity.
It corresponds to the invariant state of the map.
The dynamical properties of the map can be described by the
modulus of the subleading eigenvalue,
which determines the behaviour of the system under consecutive actions of the map.

For a generic random map, distributed according to the flat
(Euclidean) measure in the space of quantum operations,
the subleading eigenvalue is almost surely strongly smaller than one,
which implies that the convergence to equilibrium occurs
exponentially fast.
In fact statistical properties of the evolution matrix $\Phi$
of a random operation can be described by  the real Ginibre ensemble \cite{BCSZ09}.
The same is true for operations corresponding to
quantum dynamical systems under assumption of
classical chaos and strong decoherence  \cite{BSCSZ10,RSZ11}.
The real ensemble is applicable here since the map  $\Phi$ preserves Hermiticity
and the matrix $\Phi$ becomes real in the generalised Bloch vector representation
\cite{BZ06}.

Observe that independently of the dimension $k$ of the environment
the dimension $d$ of the output state $\rho'$ does not change under the action of the
transformation (\ref{env_rep}), so that $\Phi: {\cal M}_d \to {\cal M}_d$.
 However, in several applications one uses
a more general class of  quantum maps, which may change the dimension of the state.
Among them one distinguishes an important class of {\sl complementary maps}.
A map $\tilde \Phi$ complementary to $\Phi$ can be defined as  \cite{Ho05,king}
\begin{equation}
\sigma= {\tilde \Phi}(\rho)={\rm tr_{S}}\left[U(\rho\otimes |\nu\rangle \langle \nu|)U^{\dagger}\right].
\label{compl}
\end{equation}
Note that the only difference with respect to (\ref{env_rep}) is that
the partial trace is performed with respect to the principal system $S$.
The output state $\sigma \in {\cal M}_k$ describes thus the final
state of the environment after the interaction with the system described
by the global unitary matrix $U$. Thus the complementary map
sends the density matrix of size $d$ into a state of size $k$,
   ${\tilde \Phi}: {\cal M}_d \to {\cal M}_k$.
Comparing both maps we see that the role of parameters $d$ and $k$
is interchanged: the rank of the dynamical matrix corresponding to $\Phi$
is equal to $k$, while for the complementary map $\tilde \Phi$
it is equal to $d$.

The complementary map $\tilde \Phi$ can also be described in the form
(\ref{map2}). However, as the dimension $d$ of the input state $\rho$ and the
dimension $k$ of the output state  $\sigma$ do differ,
the matrix $\Phi_c$  representing the evolution operator ${\tilde \Phi}$
 under the action of
a complementary map is a rectangular matrix of size $k^2\times d^2$.
To study spectral properties  of such evolution operators
one has therefore to go beyond the standard
diagonalisation procedure which holds for square matrices only.

\subsection{Complementary maps and real induced Ginibre ensemble}
\label{w:section4}

We are now in position to investigate spectral properties
of evolution operators associated with a random complementary map
defined by (\ref{compl}).
As such an operator is represented by a rectangular matrix,
we use formula (\ref{atilda}) to obtain a square matrix,
which can be diagonalized in the standard way.

Random complementary maps were generated numerically by selecting
a fixed initial state
$|\nu\rangle$ of the $k$ dimensional environment and by plugging into
(\ref{compl}) a random unitary matrix $U$ distributed according to the Haar measure on
$U(kd)$. This procedure allows one to generate random maps according to the
Euclidean measure in the space of all quantum operations  \cite{BCSZ09}.

Fig.~\ref{real_aschlange} shows the exemplary
spectra of ${\tilde \Phi}$ associated with random complementary maps,
which transform an initial quantum state of size $d=14$
into an output state of various dimensionality.
In the symmetric case $k=d$ the evolution operator
is represented by a square matrix, so the spectrum covers the entire disk.
In the cases $k\ne d$ the evolution operator $\tilde \Phi$
is represented by a rectangular matrix $\Phi_c$.
Thus the spectrum of its quadratisation $\tilde{ \tilde{ \Phi}}$,
which corresponds to a matrix from the induced Ginibre ensemble,
covers asymptotically a ring in the complex plane.
The radius of the inner ring depends on the difference $|k-d|$.\\
%

\begin{figure}[]
\centering
 \includegraphics[width=5.0in]{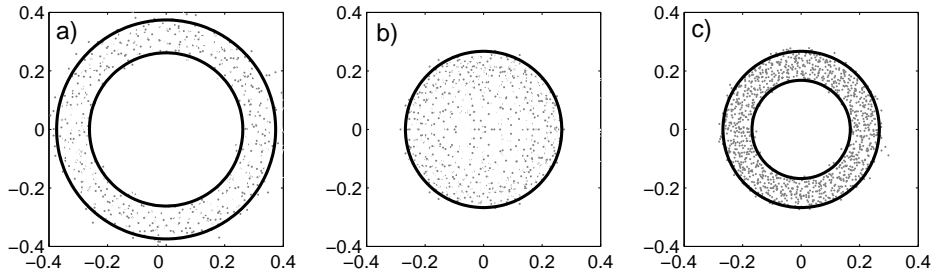}
\caption{
Spectra of operators ${\tilde {\tilde \Phi}}$ associated with evolution operators
  for generic complementary quantum maps for fixed
size $d=14$ of the initial system and different  sizes of the environment, a) $k=10$,
 b) $k=14=d$ (which implies no hole in the support of the spectrum),
and c) $k=18$. The data in each panel are superimposed
  from $8$ realisations of random complementary operations.
  Solid circles represent the radii implied by Eq. (\ref{ring2}).
}
\label{real_aschlange}
\end{figure}

Consider a random complementary map ${\tilde \Phi}$,
which sends the set ${\cal M}_d$ of $d$-dimensional states
into ${\cal M}_k$. The evolution operator
can be represented by a rectangular matrix,
$\Phi_c$ of size $k^2\times d^2$.
Assume first that $k\ge d$.
Constraints imposed by the trace preserving condition
(\ref{trace2}) are known to be weak even for relatively small dimensions
\cite{BCSZ09,BSCSZ10}, so it is legitimate to assume that
$\Phi_c$ can be described by a real random rectangular matrix
of the Ginibre ensemble. Thus a quadratisation $\tilde{\tilde{\Phi}}$
of $\Phi_c$ will be described by an appropriately rescaled
matrix of the induced Ginibre ensemble with $M=k^2$ and $N=d^2$.

To find the scaling factor we need to discuss the normalisation of
$\Phi_c$. The squared norm,
$||\Phi_c||^2={\rm Tr} \Phi_c^{\dagger} \Phi_c$
is equal to the norm $||D||^2={\rm Tr} D^2$
 of the Hermitian dynamical matrix, which contains the same entries
in a different order \cite{BZ06}.
Since the map $\tilde \Phi$ changes the dimension $d$ of the input state
into $k$, expression (\ref{dyn}) implies that the corresponding
dynamical matrix $D$ is a $dk\times dk$
Hermitian matrix of rank $d$.

Neglecting the constraints implied by the trace preserving condition,
Tr$_B D={\mathbbm 1}$, we assume that in case of a random complementary map
the dynamical matrix $D$ behaves as a random mixed state
 generated by an induced measure \cite{ZS01}.
In such a case the average purity of a random state $\rho$ of size  $N$
obtained  from an initially random pure state of size $NK$
by the partial trace of an over the $K$ dimensional environment reads,
 $\langle \Tr \rho^2\rangle =(K+N)/KN$. In the case of
dynamical matrix of a complementary map
we need to substitute the correct dimensions, $N\to kd$ and $K\to d$,
and use the normalisation constant $D=d\varsigma$
to obtain an approximate expression
\begin{equation}
\langle \Tr D^2\rangle = d^2 \langle \Tr \varsigma ^2\rangle
\approx d^2 \frac {kd+d}{d^2k} =  \frac {d(k+1)}{k} \approx d .
\label{normd}
\end{equation}
This in turn implies the relation
 $ \langle{\rm Tr} \left(\Phi_c^{\dagger}\Phi_{c}\right)\rangle\approx d$,
so the average norm of the superoperator
$\Phi_c$ representing the complementary channel  $\tilde \Phi$
does not depend on the parameter $k$~and  can be estimated by
 $||\Phi_c|| \approx \sqrt{d}$ -- see Fig.~\ref{mean_tr}.

\begin{figure}[ht]
  \centering
  \includegraphics[width=4.0in]{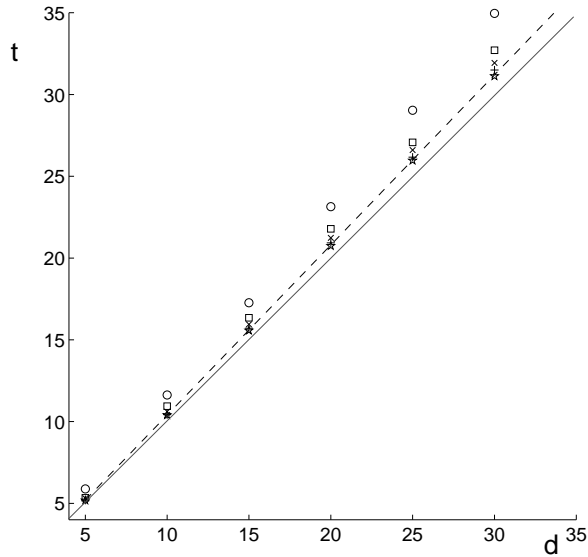}
  \caption[]{Average squared norm of the superoperator
    related to the complementary map,
    $t=\langle {\rm Tr} \Phi_c^{\dagger} \Phi_c \rangle$
as a~function of size  $d$ of the input system. Various symbols
    denote different values of the size of the output system,
     $k=5$ ($\circ$), $k=10$ ({\tiny $\square$}),
 $k=15$ ($\times$), $k=20$ ($+$) and $k=25$ ($\star$).
Dashed line is plotted to guide the eye. It corresponds to $k=25$ and shows
  that the bigger value of $k$, the better
 approximation to the asymptote represented by the solid diagonal line.}
\label{mean_tr}
\end{figure}

As discussed in previous sections, an  $M \times N$ matrix $X$
(with $N \le M$) of the Gaussian ensemble (\ref{bb0}), for which
the average norm satisfies  $\langle {\rm Tr} X^{\dagger}X \rangle = NM$,
leads to a square random matrix $G$ of size $N$ from
the induced Ginibre ensemble (\ref{bb5}) by the way of quadratisation. Its spectrum is (asymptotically)
supported in the ring $r\in [r_{\rm in}, r_{\rm out}]$
with $r_{\rm in}=\sqrt{L}$ and $r_{\rm out}=\sqrt{L+N}$.
To apply this ensemble for a rectangular superoperator
matrix $\Phi_c$ it is then sufficient to
substitute $M\to k^2$ and $N\to d^2$
and to use the  rescaling $\Phi^c=X/k\sqrt{d}$
to match the normalisation. Therefore
the spectrum of the quadratisation $\tilde{\tilde{\Phi}}$
of the superoperator $\Phi_c$ associated with the complementary map   $\tilde{\Phi}$
forms a ring of the inner radius
$r_{\rm in}=\sqrt{1-d^2/k^2}/\sqrt{d}$ while
the outer radius reads  $r_{\rm out}=1/\sqrt{d}$.


Consider now the other case, in which the dimensions determining the complementary map satisfy $k \le d$. Then the superoperator $\Phi_c$, represented by a rectangular matrix $k^2\times d^2$, gets its quadratisation by the same procedure from the transposed matrix as mentioned in  Section~\ref{w:section4.2}, and again it is described by the induced Ginibre matrices with the level density analysed in Section~\ref{sec:real}. The key difference with respect to this distribution, describing the transposed rectangular matrix $X^T$, is that the role of variables $M$ and $N$ is exchanged. Hence in the case for the transposed superoperator $\Phi_c^T$, one needs to exchange the dimensions $k$ and $d$. Taking this into account and the normalisation,  $||\Phi_c||^2 \approx d$, one infers that the outer radius of the ring reads in this case $r_{\rm out}=\sqrt{d^2}/k\sqrt{d}=\sqrt{d}/k$.

Thus the square matrices $\tilde{\tilde{\Phi}}$ associated with the superoperators corresponding to random complementary maps
${\tilde \Phi}: {\cal M}_d \to {\cal M}_k$ can be described with
real matrices of the induced Ginibre ensemble.
Apart of the leading eigenvalue $\lambda_1=1$,
which is implied by the trace preserving condition,
the spectra are asymptotically localized in a ring in the complex plane.
Both radii of the ring read
\begin{equation}
[r_{\rm in}, r_{\rm out}]= \
\left\{
\begin{array}{ll}
\left[\frac{1}{\sqrt{d}} \sqrt{1 -\frac{d^2}{k^2}} ,\
 \frac{1}{\sqrt{d}}\right]           & \ {\rm for} \ \ k \geqslant d \\

\left[\frac{\sqrt{d}}{k} \sqrt{1 - \frac{k^2}{d^2}} ,\
      \frac{\sqrt{d}}{k}\right]      &  \ {\rm for} \ \ k\leqslant d
 \ . \end{array}
\right.
\label{ring2}
\end{equation}
The above predictions show good agreement with numerical data
obtained for several realisations of random complementary maps
and shown in Fig.~\ref{real_aschlange}.
To obtain a clearer figure
the plot is magnified by a factor $\sqrt{d}$.

\section{Concluding remarks}

Although the specific example of the Ginibre ensemble is completely solved,
the theory of non--Hermitian random matrices is still far from being as thoroughly
understood as its Hermitian counterpart.
In this work we have introduced a new generalisation of the
ensemble of non--Hermitian Ginibre matrices and derived
explicit results for spectral density in the complex and real cases.
Using the method of skew-orthogonal polynomials we derived
 various spectral correlation functions which can be expressed as Pfaffians.
Analysing asymptotic behaviour in the limit of large matrix dimensions
 we have found a universal behaviour of the eigenvalue statistics.

The induced Ginibre ensemble of square matrices of size $N$
is parametrised by a single discrete parameter $L\ge 0$.
In the case $L=0$ the model reduces to the standard Ginibre ensemble \cite{Ginibre1965}, with the eigenvalues, for large $N$, uniformly distributed in the disk of radius $\sqrt{N}$ about the origin, while for $L>0$ the eigenvalues are repelled from the origin and form a ring in the complex plane $\sqrt{L} < |z| < \sqrt{N+L}$.

This form of the spectrum suggests a comparison with the model (\ref{ginibFZ}) of
Feinberg--Zee, for which the 'single ring' theorem was proven \cite{FZ97,GKZ09,GZ10}.
Our model formally belongs to the Feinberg--Zee class,
with the potential
\begin{equation}
     V(G{^\dagger}G)=  G{^\dagger}G - L \log G{^\dagger}G ,
\label{logpot}
\end{equation}
 but due to the log function the assumption
that the potential is polynomial is not satisfied.
Analysing the joint probability distribution for induced Ginibre ensemble
it is easy to see that eigenvalues near the origin are unlikely,
which implies the ring of eigenvalues.
On the other hand it is not simple to find a specific mechanism responsible for
the single ring distribution for the general version of the Feinberg--Zee model.
Furthermore,  in the latter model the eigenvalue correlation functions are not known,
while we could compute them for the induced Ginibre ensemble
for the complex and real versions of the model.

Our work leads to a straightforward explicit algorithm
to generate random matrices from the induced Ginibre ensemble.
Its is sufficient to take a rectangular  $M\times N$  random Gaussian matrix
$X$, construct the positive Wishart like matrix, $X^{\dagger}X$ of size $N$,
take its square root and multiply it by a random unitary $U$ distributed according to the
Haar measure on the unitary group $U(N)$.
The result $G=U\; \sqrt{X^{\dagger}X}$ is distributed according to
the desired joint probability function (\ref{ginibind}),
while its singular values are described by the Marchenko--Pastur distribution
with parameter $c=M/N$. This follows from the fact
that  $G^{\dagger}G=X^{\dagger}X$, hence the square random matrix $G$ has the same singular values as the initial rectangular Gaussian matrix $X$.

Alternatively, square random matrices from the induced Ginibre
ensemble can be obtained by the quadratisation (\ref{WX})--(\ref{bb11}), (\ref{atilda}) of rectangular Gaussian matrices. Although for a given rectangular matrix
various quadratisation algorithms produce different square matrices,
the statistical properties of square matrices obtained by quadratisation
of rectangular random Gaussian matrices do not depend on the choice of the algorithm.

An analogous construction works also in the real case.
Take a rectangular  $M\times N$ real  Ginibre matrix
$X$, construct the square root of the positive Wishart matrix, $X^{T}X$
and multiplying it by a random orthogonal matrix $O$ distributed according to the
Haar measure on the orthogonal group $O(N)$.
The result $G=O\sqrt{X^{T}X}$ is then distributed according to
the induced ensemble of real Ginibre matrices, specified by (\ref{ginibind}) with $\beta=1$.

The induced ensemble of random Ginibre matrices offers a simple model for
further research within the developing theory of non--Hermitian random matrices.
We are tempted to believe it will find its applications in several fields of physics.
For instance they can be helpful for analysis of evolution operators
associated to certain classes of generic quantum maps.
In the simplest case as the dimension of the input and output states are the
same, the associated evolution operator can be described
by  the standard ensemble of real Ginibre matrices \cite{BSCSZ10}.
However,  studying complementary quantum operations, or other quantum maps
in which the dimensions of the initial and the final states do differ,
one copes with evolution operators represented by rectangular matrices.
Our work provides evidence that the statistical properties of
square matrices associated with the evolution operators for complementary quantum maps
do correspond to the induced ensemble of real Ginibre matrices.

%

\medskip

\ack

We would like to thank Gernot Akemann, Marek  Bo{\.z}ejko, Yan Fyodorov, Dmitry Savin  and an anonymous referee  for their useful comments; one of the referee's comments is reproduced in the footnote following equation (\ref{bb11}).
Financial support by the SFB Transregio-12 project
der Deutschen Forschungsgemeinschaft and the grant financed by the Polish National Science Centre under the contract number  DEC-2011/01/M/ST2/00379 is gratefully acknowledged.

\appendix

\section{Proof of Lemma \ref{lem1}}
\label{appendix_A}

Assuming $W$ as in (\ref{bb11}), by multiplying through in  (\ref{WX}), one obtains an equation for $C$,
%
$C^{\dagger}Y+\left(\mathbbm{1}_{M-N}-C^{\dagger}C\right)^{1/2}Z=0$.
Hence,
\begin{equation}
%
Z=-\left(\mathbbm{1}_{M-N}-C^{\dagger}C\right)^{-1/2}C^{\dagger}Y=
-C^{\dagger}\left(\mathbbm{1}_N-CC^{\dagger}\right)^{-1/2}Y\; .
\end{equation}
Consequently, by making use of (\ref{WX}) again,
\begin{equation}
G=\left(\mathbbm{1}_N-CC^{\dagger}\right)^{1/2}Y-CZ
=\left(\mathbbm{1}_N-CC^{\dagger}\right)^{-1/2}Y.
\label{aa2}
\end{equation}
It is easy to check that
\begin{equation}
Z^{\dagger}Z =
 Y^{\dagger}\left(\mathbbm{1}_N-CC^{\dagger}\right)^{-1}Y-Y^{\dagger}Y
\label{aa3}
\end{equation}
which implies that
\begin{equation}
Y^{\dagger}Y+Z^{\dagger}Z=Y^{\dagger}\left(\mathbbm{1}_N-CC^{\dagger}\right)^{-1}Y .
\label{aa4}
\end{equation}
This in turn allows us to write
\begin{equation}
\left(\mathbbm{1}_N-CC^{\dagger}\right)^{-1}=
\frac{1}{Y^{\dagger}}\left(Y^{\dagger}Y + Z^{\dagger}Z\right)\frac{1}{Y} ,
\label{aa5}
\end{equation}
which when substituted into (\ref{aa2}) yields
(\ref{atilda}).
Note that the desired result (\ref{atilda})
can be also rewritten in a more symmetric form,
\begin{equation}
G=
Y\left(Y^{\dagger}Y\right)^{-1/2}\left(\mathbbm{1}_N+\left(Y^{\dagger}Y\right)^{-1/2}
       Z^{\dagger}Z\left(Y^{\dagger}Y\right)^{-1/2}\right)^{1/2}
 \left(Y^{\dagger}Y\right)^{1/2},
\label{atilda1}
\end{equation}
which shows that all matrix square roots
operate correctly on positive definite objects.

\section{Proof of Theorem \ref{Pfaffian1}}
\label{appendix3}
In this appendix we determine the finite $N$ correlation functions for the real induced Ginibre by determining the matrix entries of the Pfaffian kernel
given in \ref{Pfaffianentries}
and hence prove Theorem \ref{Pfaffian1}. Equipped with the appropriate skew-orthogonal polynomials and their normalisation the task of determining the entries of the matrix kernel for the $(j,m)-$correlation functions can now proceed. Firstly \eref{kerneleq} implies for $w,z\in\mathbb{C}$:
\be
\fl DS_N(z,w)=\frac{1}{\sqrt{2\pi}}\psi(w)\psi(z)(z-w)
\sum_{j=0}^{N-2}\frac{\big(wz)^{j+L}}{\Gamma(L+j+1)}.
\ee
Noting that $S_N(z,w)=iDS_N(z,\bar{w})$ and $IS_N(z,w)=-DS_N(\bar{z},\bar{w})$ we have completely determined the entries of the complex-complex matrix kernel.\\
\\
Let us next consider the case $x\in \mathbb{R}$, $z\in\mathbb{C}$. The following approach is borrowed from \cite{FN}. We observe that
\bes
q_{2j+1}(x)=-e^{\f x^2}x^{-L}\frac{\partial}{\partial x}\left[e^{-\f x^2}x^{2j+L}\right],
\ees
which implies for $j>0$:
\bes
\tau_{2j+1}(x)=e^{-\f x^2}x^{j+L}
\ees
Furthermore direct computation shows that:
\bes
\tau_{1}(x)-\frac{L}{2}\int_{\mathbb{R}}\sgn(x-t)e^{-\f t^2}x^{L-1}dt=e^{-\f x^2}x^{2j+L}
\ees
All in all
\begin{eqnarray*}
\fl\sum_{j=0}^{\frac{N}{2}-1}\frac{1}{r_j}\tilde{q}_{2j}(z)\tau_{2j+1}(x)
=\frac{1}{\sqrt{2\pi}}\psi(z)e^{-\f x^2}\sum_{j=1}^{\frac{N}{2}-1}\frac{(xz)^{L+2j}}{\Gamma(L+2j+1)}\\
\fl+\frac{1}{\sqrt{2}\pi\Gamma(L+1)}\psi(z)z^L
\left[\tau_1(x)-\frac{L}{2}\int_{\mathbb{R}}\sgn(x-t)e^{-\f t^2}x^{L-1}dt\right]\\
\fl+\frac{1}{\sqrt{2\pi}\Gamma(L+1)}\psi(z)z^L\frac{L}{2}\int_{\mathbb{R}}\sgn(x-t)e^{-\f t^2}x^{L-1}dt\\
\fl=\frac{1}{\sqrt{2\pi}}e^{-\f x^2-\f z^2}\sqrt{\erfc(\sqrt{2}\Image(z))}\sum_{j=0}^{\frac{N}{2}-1}\frac{(xz)^{L+2j}}{\Gamma(L+2j+1)}
+\frac{1}{\sqrt{2\pi}}\psi(z)z^L2^{\frac{L}{2}-1}\frac{\Gamma(\frac{L}{2},\f x^2)}{\Gamma(L+1)}
\end{eqnarray*}
In addition to that:
\begin{eqnarray*}
 S&=\sum_{j=0}^{\frac{N}{2}-1}\frac{1}{r_j}\tilde{q}_{2j+1}(z)\tau_{2j}(x)\\
&=\frac{1}{\sqrt{2\pi}}\psi(z)z^L\sum_{j=1}^{\frac{N}{2}-1}\frac{\left[z^{2j+1}-(L+2j)z^{2j-1}\right]\tau_{2j}(x)}{\Gamma(L+2j+1)}
+\frac{1}{\sqrt{2\pi}\Gamma(L+1)}\psi(z)z^L\tau_0(x)
\end{eqnarray*}
Rearranging the summation gives:
\begin{eqnarray*}
S&=\frac{1}{\sqrt{2\pi}\Gamma(L+N-1)}\psi(z)z^{L+N-1}\tau_{N-2}(x)\\
&\quad-\frac{1}{\sqrt{2\pi}}\psi(z)z^L\sum_{j=0}^{\frac{N}{2}-2}\frac{\left[\tau_{2j+2}(x)-(L+j+1)\tau_{2j}(x)\right]z^{2j+1}}{\Gamma(L+2j+1)}
\end{eqnarray*}
Another differential equation:
\bes
q_{2j+2}(x)-(2j+L+1)q_{2j}(x)=-e^{\f x^2}x^{-L}\frac{\partial}{\partial x}\left[e^{-\f x^2}x^{2j+L+1}\right],
\ees
leads to:
\bes
\tau_{2j+2}(x)-(2j+L+1)\tau_{2j}(x)=e^{-\f x^2}x^{2j+L+1}
\ees
As a consequence
\begin{eqnarray*}
\fl S=-\frac{1}{\sqrt{2\pi}}\psi(z)z^{L+N-1}2^{\frac{L}{2}+N-\frac{3}{2}}\sgn(x)\frac{\gamma(\frac{L}{2}+\frac{N}{2}-\f,\f x^2)}{\Gamma(L+N-1)}-\frac{1}{\sqrt{2\pi}}\psi(z)e^{-\f x^2}\sum_{j=0}^{\frac{N}{2}-2}\frac{(xz)^{L+2j+1}}{\Gamma(L+2j+2)}
\end{eqnarray*}
Finally we obtain:
\begin{eqnarray*}
 \fl S_N(z,w)&=\frac{1}{\sqrt{2\pi}}e^{-\f x^2}\psi(z)\sum_{j=0}^{N-2}\frac{(xz)^{L+2j}}{\Gamma(L+2j+1)}\\
\fl&+\frac{1}{\sqrt{2\pi}}\psi(z)z^{L+N-1}2^{\frac{L}{2}+N-\frac{3}{2}}\sgn(x)\frac{\gamma(\frac{L}{2}+\frac{N}{2}-\f,\f x^2)}{\Gamma(L+N-1)}+\frac{1}{\sqrt{2\pi}}\psi(z)2^{\frac{L}{2}-1}\frac{\Gamma(\frac{L}{2},\f x^2)}{\Gamma(L+1)}
\end{eqnarray*}
The last entry requiring explicit computation is $IS_N(x,y)$ for $x,y\in \mathbb{R}$. In this case the relationship:
\bes
IS_N(x,y)=-\int_x^yS_N(t,y)dt
\ees
comes handy. Using the expression obtained for $S_N(x,y)$ and in addition to that employing the integral representation
\bes
 e^{-ty}\sum_{j=0}^{N-2}\frac{(ty)^{L+2j}}{\Gamma(L+2j+1)}=
\left[\frac{\gamma(L,ty)}{\Gamma(L)}-\frac{\gamma(L+N-1,ty)}{\Gamma(L+N-1)}\right]
\ees
leads to the following starting point for our derivation
\begin{eqnarray*}
&IS_N(x,y)=-\frac{1}{\sqrt{2\pi}\Gamma(L)}\int_x^ye^{-\f(t-y)^2}(y-t)\gamma(L,ty)dt\\
&\quad+\frac{1}{\sqrt{2\pi}\Gamma(L+N-1)}\int_x^ye^{-\f(t-y)^2}(y-t)\gamma(L+N-1,ty)dt\\
&\quad-\frac{1}{\sqrt{2\pi}}\sgn(x)2^{\frac{L}{2}+N-\frac{3}{2}}\frac{\gamma(L+\frac{N}{2}-\f,\f x^2)}{\Gamma(L+N-1)}
\int_x^ye^{\f t^2}t^{L+N-1}dt\\
&\quad-\frac{1}{\sqrt{2\pi}}\frac{\Gamma(\frac{L}{2},\f x^2)}{\Gamma(L+1)}\int_x^ye^{\f t^2}t^{L}dt\;.
\end{eqnarray*}
The above expression can be simplified by employing integration by parts with respect to $t$. As a conclusion we have derived all the possible entries of the Pfaffian matrix kernel.

\section{Asymptotic Analysis of Correlation Functions}
\label{appendix_d}
In this appendix we state the asymptotic behaviour of the correlation functions of the complex induced Ginibre ensemble for different scaling regimes.
\begin{thm}[The limiting correlation functions in the bulk]\label{complexcorr1}
Let $u,z_1,\ldots,z_n$ be complex numbers and set $\la_k=\sqrt{N}u+z_k$ for $k=1,\ldots,n$ and $L=N\al$, then for $u\in R$:
\bes
\lim_{N\to\infty}R_n(\la)=\det{\left[\frac{1}{\pi}\exp\Big(-\frac{|z_j|}{2}-\frac{|z_k|}{2}+z_j\bar{z}_k\Big)\right]}_{j,k=1}^n
\ees
with $R=\{r\in\mathbb{C}|\sqrt{\al}\leq |r|\leq \sqrt{\al+1}\}$.
\end{thm}
\begin{proof}
We need to analyze the asymptotics of the kernel:
\begin{eqnarray*}
\fl K_N(z_j,z_k)&=\frac{1}{\pi}e^{-\f\big|\sqrt{N}u+z_j\big|\big|\sqrt{N}\bar{u}+\bar{z}_j\big|-
\f\big|\sqrt{N}u+z_k\big|\big|\sqrt{N}\bar{u}+\bar{z}_k\big|}\\
\fl&\quad\times\sum_{l=0}^{N-1}\frac{\Big(\big|\sqrt{N}u+z_j\big|\big|\sqrt{N}\bar{u}+\bar{z}_k\big|\Big)^{l+L}}{\Gamma(l+L+1)}\\
&=\frac{1}{\pi}e^{-N|u|^2-\sqrt{N}(u\bar{z}_k+\bar{u}z_j)-z_j\bar{z}_k}
e^{-\frac{|z_j|^2}{2}-\frac{|z_k|^2}{2}+z_j\bar{z}_k} e^{\frac{\sqrt{N}}{2}(\bar{u}z_j-u\bar{z}_j)}e^{-\frac{\sqrt{N}}{2}(\bar{u}z_k-u\bar{z}_k)}\\
\fl&\quad\times\sum_{l=0}^{N-1}\frac{\Big(N|u|^2+\sqrt{N}(u\bar{z}_k+\bar{u}z_j)+z_j\bar{z}_k\Big)^{l+L}}{\Gamma(l+L+1)}.
\end{eqnarray*}
In order to simplify this expression we resort to the following trick outlined in \cite{BS}. We set: $\psi_N(z_j)=e^{\frac{\sqrt{N}}{2}(\bar{u}z_j-u\bar{z}_j)}$ and define the diagonal matrix
\bes
D_n=\diag\Big(\psi_N(z_1),\psi_N(z_2),\ldots,\psi_N(z_n)\Big).
\ees
We note the following: $\left|\psi_N(s)\right|=1$ and as a consequence
\bes
R_n(\la)=\det\big(D_n\big)R_n(\mathbf{\la})\det\big(D_n^{\dagger}\big).
\ees
Thus we can write for the $n$-point correlation function
\begin{eqnarray*}
\fl R_n(\la)=\det\Big[\frac{1}{\pi}e^{-\frac{|z_j|}{2}-\frac{|z_k|}{2}+z_j\bar{z}_k}
e^{-N|u|^2-\sqrt{N}(u\bar{z}_k+\bar{u}z_j)-z_j\bar{z}_k}\\
\times\sum_{l=0}^{N-1}\frac{\Big(N|u|^2+\sqrt{N}(u\bar{z}_k+\bar{u}z_j)+z_j\bar{z}_k\Big)^{l+L}}{\Gamma(l+L+1)}\Big]_{j,k=1}^n.
\end{eqnarray*}
Applying another saddle-point analysis it can be shown that for $u\in R$:
\bes
\fl\lim_{N\to\infty} e^{-N|u|^2-\sqrt{N}(u\bar{z}_k+\bar{u}z_j)-z_j\bar{z}_k}
\sum_{l=0}^{N-1}\frac{\Big(N|u|^2+\sqrt{N}(u\bar{z}_k+\bar{u}z_j)+z_j\bar{z}_k\Big)^{l+L}}{\Gamma(l+L+1)}=1,
\ees
which gives the limiting expression for the $n$-point correlation functions.
\end{proof}
Employing the same technique again, albeit with a slightly different saddle-point method, yields the limiting behavior of the $n$-point correlation functions around the circular edges of the eigenvalue density.
\begin{thm}[The limiting correlation functions at the edges]
Let $u,z_1,\ldots,z_n$ be complex numbers with $|u|=1$,\\
1. Setting $\la_k=\sqrt{N(\al+1)}u+z_k$ for $k=1,\ldots,n$ leads to the limiting correlation functions at the outer edge $\sqrt{N(\al+1)}$:
\bes
\fl\lim_{N\to\infty}R_n(\la_1,\ldots,\la_n)=
\det{\left[\frac{1}{2\pi}\exp\Big(-\frac{|z_j|^2}{2}-\frac{|z_k|^2}{2}+z_j\bar{z}_k\Big)
\Big(\erfc\big(\frac{z_j\bar{u}+\bar{z}_ku}{\sqrt{2}}\big)\Big)\right]}_{j,k=1}^n.
\ees
The same limiting expression is found around the inner edge $\sqrt{N\al}$ of the eigenvalue density by setting $\la_k=\sqrt{N\al}u-z_k$ for $k=1,\ldots,n$.
\end{thm}

\section{Limiting correlation functions for the real induced Ginibre ensemble}

Throughout this section it is assumed that $L=N\al$.

\label{appen_C}
\begin{thm}[The limiting correlation functions in the bulk]\label{complexcorr}
Let $u\in\mathbb{R}$ such that $\sqrt{\al}<|u|<\sqrt{\al+1}$ and let $r_1,\ldots,r_j\in\mathbb{R}$ as well as $s_1,\ldots,s_m\in \mathbb{C}_+\backslash\mathbb{R}$. Furthermore set $x_t=\sqrt{N}u+r_t$ for $t=1,\ldots,j$, $z_v=\sqrt{N}u+s_v$ for $v=1,\ldots,m$ and $L=N\al$, then:
\bes
\fl\lim_{N\to\infty}R_{j,m}(x_1,\ldots,x_j,z_1,\ldots,z_m)=
\Pfaff\left[\begin{array}{cc}
        K(r_t,r_{t'}) & K(r_t,s_{v'}) \\
        K(s_v,r_{t'}) & K(s_v,s_{v'})
      \end{array}\right]
\ees
where $t,t'=1,\ldots,j$ and $v,v'=1,\ldots, m$.
\begin{enumerate}
  \item The limiting real/real kernel is given by
    \bes
   \fl K(r,r')=\left[
              \begin{array}{cc}
                \frac{1}{\sqrt{2\pi}}(r-r')e^{-\f(r-r')^2} & \frac{1}{\sqrt{2\pi}}e^{-\f(r-r')^2} \\
                -\frac{1}{\sqrt{2\pi}}e^{-\f(r-r')^2} & \f\sgn(r-r')\erfc\left(\frac{|r-r'|}{\sqrt{2}}\right) \\
              \end{array}
            \right].
    \ees
  \item The limiting complex/complex kernel is given by
\begin{eqnarray*}
    \fl K(z,z')=\frac{1}{\sqrt{2\pi}}\sqrt{\erfc(\sqrt{2}\Image(z))\erfc(\sqrt{2}\Image(z'))}\\
\times\left[              \begin{array}{cc}
               (z'-z)e^{-\f(z-z')^2} & i(\bar{z}-z')e^{-\f(z-\bar{z'})^2} \\
                i(z'-\bar{z})e^{-\f(\bar{z}-z')^2}  & (\bar{z}-\bar{z'})e^{-\f(\bar{z}-\bar{z'})^2} \\
              \end{array}
            \right].
    \end{eqnarray*}
  \item The limiting real/complex kernel is given by
\bes
    \fl K(r,z)=\frac{1}{\sqrt{2\pi}}\sqrt{\erfc(\sqrt{2}\Image(z))}\left[
              \begin{array}{cc}
                (z-r)e^{-\f(r-z)^2} & i(\bar{z}-r)e^{-\f(r-\bar{z})^2} \\
                -\e^{-\f(x-z)^2} & -ie^{-\f(x-\bar{z})^2} \\
              \end{array}
            \right].
    \ees
\end{enumerate}
\end{thm}
\begin{proof}
In the following only the derivation of the limiting behaviour of the real/real kernel will be outlined, as this is the most involved computation. All other results can be deduced in a similar fashion. We start by computing the asymptotic behaviour of $S_N$. Using the results from \ref{complexcorr1} it is obvious that:
\bes
\lim_{N\to\infty}s_N(\sqrt{N}u+r,\sqrt{N}u+r')=\frac{1}{\sqrt{2\pi}}e^{-\f(r-r')^2}.
\ees
The next term to be analyzed is:
\beas
\fl r_N(\sqrt{N}u+r,&\sqrt{N}u+r')=\frac{1}{\sqrt{2\pi}}\sgn(\sqrt{N}u+r)\frac{2^{\frac{N}{2}(\al+1)-\frac{3}{2}}}{\Gamma(N(\al+1)-1)} \\
&\times e^{-\f (\sqrt{N}u+r')^2}(\sqrt{N}u+r')^{N(\al+1)-1}\gamma\left(\frac{N}{2}(\al+1)-\f,\f (\sqrt{N}u+r^2)\right)
\eeas
We can apply the duplication formula for the gamma function:
\bes
\Gamma(2z)=\Gamma(z)\Gamma(z+\f)2^{2z-1}\frac{1}{\sqrt{\pi}}.
\ees
and obtain:
\beas
\fl r_N(\sqrt{N}u+r,&\sqrt{N}u+r')=\sgn(\sqrt{N}u+r)2^{\frac{N}{2}(\al+1)-\frac{3}{2}}e^{-\sqrt{N}ur'}(1+\frac{r'}{\sqrt{N}u})^{N(\al+1)-1}\\
&\times\frac{[Nu^2]^{\frac{N}{2}(\al+1)-\f}}{\Gamma(\frac{N}{2}(\al+1))\Gamma(\frac{N}{2}(\al+1-\f))}\gamma\left(\frac{N}{2}(\al+1)-\f,\f (\sqrt{N}u+r^2)\right).
\eeas
Furthermore the use of the Stirling formula:
\bes
\Gamma(z)\sim e^{-x}x^x\sqrt{\frac{2\pi}{x}},
\ees
as well as:
\bes
(1+\frac{r'}{\sqrt{N}u})^{N(\al+1)-1}\sim e^{\sqrt{N}(\al+1)\frac{s}{u}-\frac{s^2}{2u^2}},
\ees
leads to:
\beas
r_N(\sqrt{N}u+r,\sqrt{N}u+r')\sim\sgn(u)\frac{1}{\sqrt{\pi}}e^{\sqrt{N}s(\frac{\al+1-u^2}{u})}e^{-\f s^2(\frac{\al+1+u^2}{u^2})}\\
e^{\frac{N}{2}(\al+1-u^2+2(\al+1)\ln(\frac{u^2}{\al+1}))}\frac{1}{\Gamma(\frac{N}{2}(\al+1-\f)}\gamma\left(\frac{N}{2}(\al+1)-\f,\f (\sqrt{N}u+r^2)\right).
\eeas
Now it can easily be seen that:
\bes
\lim_{N\to\infty}e^{\sqrt{N}s(\frac{\al+1-u^2}{u})}e^{-\f s^2(\frac{\al+1+u^2}{u^2})}=0.
\ees
In addition to that for $|u|^2<\al+1$  the expression $\al+1-u^2+2(\al+1)\ln(\frac{u^2}{\al+1})$ is negative and a saddle point method shows:
\bes
\lim_{N\to\infty}\frac{1}{\Gamma(\frac{N}{2}(\al+1-\f)}\gamma\left(\frac{N}{2}(\al+1)-\f,\f (\sqrt{N}u+r^2)\right)= 0.
\ees
As a conclusion we have shown that $\lim_{N\to\infty}r_N(\sqrt{N}u+r,\sqrt{N}u+r')=0$. The next term
\bes
\fl t(\sqrt{N}u+r,\sqrt{N}u+r')=\frac{1}{\sqrt{2\pi}}\frac{2^{\frac{N\al}{2}-1}}{\Gamma(N\al)}e^{-\f (\sqrt{N}u+r')^2}(\sqrt{N}u+r')^{N\al}\Gamma(\frac{N}{2}\al,\f (\sqrt{N}u+r)^2)
\ees
can be dealt with in a similar fashion. We have now determined the scaling limits of the real/real entries $S_N$ and $DS_N$. The scaling limit for the entry $IS_N$ can be found by applying a saddle point method on each of the eight integrals. In addition to that the asymptotic relationships derived for $r_N$ and $t$ can be applied.
\end{proof}

Similar calculations lead to the conclusion that in the complex bulk and also at the edges the eigenvalue correlation functions in the real induced Ginibre are exactly the same as those in the real Ginibre ensemble. We omit the derivations and only state the results in the two theorems below.

\begin{thm}[The limiting correlation functions in the complex bulk]
Let u be a complex number such that $u\in R=\{r\in\mathbb{C}|\sqrt{\al}\leq |r|\leq \sqrt{\al+1}\}$ and let $s_1,\ldots,s_m\in \mathbb{C}$.  Furthermore set $z_j=\sqrt{N}u+s_j$ for $j=1,\ldots,m$ and $L=N\al$, then for $u\in R$:
\bes
\fl\lim_{N\to\infty}R_{0,m}(-,z_1,\ldots,z_m)=
\frac{1}{\pi}\det{\left[\exp\Big(-\frac{|s_k|}{2}-\frac{|s_{k'}|}{2}+z_k\bar{z}_{k'}\Big)\right]}_{k,k'=1,\ldots,m}.
\ees
\end{thm}
\begin{thm}[The limiting correlation functions at the edges]
Let $u=\pm1$, $r_1,\ldots,r_l\in\mathbb{R}$ as well as $s_1,\ldots,s_m\in \mathbb{C}_+$. Setting $x_j=\sqrt{N(\al+1)}u+r_j$ for $t=1,\ldots,l$ and $z_k=\sqrt{N(\al+1)}u+s_k$ for $k=1,\ldots,m$ leads to Ginibre behavior for the limiting correlation functions at
the outer edge of the eigenvalue distribution as described in \cite{BS}.
In addition at the inner circular edge $x_j=\sqrt{N\al}u-r_j$ for $t=1,\ldots,l$ and $z_k=\sqrt{N\al}u-s_k$ for $k=1,\ldots,m$ the Ginibre limiting correlation functions can again be recovered.
\end{thm}
\bigskip


\bibliographystyle{unsrt}
\bibliography{induce}{}

\begin{thebibliography}{10}

\bibitem{Ginibre1965}
J.~Ginibre.
\newblock Statistical ensembles of complex, quaternion, and real matrices.
\newblock {\em Journal of Mathematical Physics}, 6:440--449, 1965.

\bibitem{Ha10}
F.~Haake.
\newblock {\em Quantum Signatures of Chaos, 3rd ed.}
\newblock Springer Berlin / Heidelberg, 2010.

\bibitem{Be97}
C.~W.~J. Beenakker.
\newblock Random-matrix theory of quantum transport.
\newblock {\em Reviews of Modern Physics}, 69(3):731--808, 1997.

\bibitem{Ef97}
K.~B. Efetov.
\newblock Quantum disordered systems with a direction.
\newblock {\em Physical Review B}, 56(15):9630--9648, 1997.

\bibitem{BP}
J.~P. Bouchaud and M.~Potters.
\newblock {\em Theory of Financial Risk and Derivative Pricing: From
  Statistical Physics to Risk Management 2nd ed.}
\newblock Cambridge University Press, Cambridge, 2004.

\bibitem{KDGO06}
J.~Kwapie{\'n}, S.~Dro{\.z}d{\.z}, A.~Z. G{\'o}rski, and P.~O{\'s}wi{\c
  e}cimka.
\newblock Asymmetric matrices in an analysis of financial correlations.
\newblock {\em Acta Phys. Pol. B}, 37(11):3039--3048, 2006.

\bibitem{KDI00}
J.~Kwapie{\'n}, S.~Dro{\.z}d{\.z}, and A.~A. Ioannides.
\newblock Temporal correlations versus noise in the correlation matrix
  formalism: An example of the brain auditory response.
\newblock {\em Physical Review E}, 62(4):5557--5564, 2000.

\bibitem{Se03}
P.~{\v S}eba.
\newblock Random matrix analysis of human eeg data.
\newblock {\em Phys. Rev. Lett.}, 91(19):198104, 2003.

\bibitem{BSCSZ10}
W.~Bruzda, M.~Smaczynski, V.~Cappellini, H.-J. Sommers, and K.~{\.Z}yczkowski.
\newblock Universality of spectra for interacting quantum chaotic systems.
\newblock {\em Physical Review E}, 81(6):066209, 2010.

\bibitem{TV04}
A.~M. Tulino and S.~Verdu.
\newblock {\em Random matrix theory and wireless communications}.
\newblock Hannover MA: now Publishers Inc., 2004.

\bibitem{Ti02}
M.~Timme, F.~Wolf, and T.~Geisel.
\newblock Coexistence of regular and irregular dynamics in complex networks of
  pulse-coupled oscillators.
\newblock {\em Physical Review Letters}, 89(25):258701, 2002.

\bibitem{S}
C.D. Sinclair.
\newblock Averages over {G}inibre{'}s ensemble of random real matrices.
\newblock {\em International Mathematics Research Notices}, 2007:rnm015, 2007.

\bibitem{LS}
N.~Lehmann and H.-J. Sommers.
\newblock Eigenvalue statistics of random real matrices.
\newblock {\em Phys. Rev. Lett.}, 67:941--944, 1991.

\bibitem{E}
A.~Edelman.
\newblock The probability that a random real {G}aussian matrix has k real
  eigenvalues, related distributions, and the circular law.
\newblock {\em Journal of Multivariate Analysis}, 60(2):203--232, 1997.

\bibitem{AK}
G.~Akemann and E.~Kanzieper.
\newblock Integrable structure of {G}inibre's ensemble of real random matrices
  and a {P}faffian integration theorem.
\newblock {\em J. Stat. Phys.}, page 1159, 2007.

\bibitem{FN}
P.~J. Forrester and T.~Nagao.
\newblock Eigenvalue statistics of the real ginibre ensemble.
\newblock {\em Physical Review Letters}, 99(5):050603, 2007.

\bibitem{Forrester2008}
P.J. Forrester and T.~Nagao.
\newblock Skew orthogonal polynomials and the partly symmetric real {G}inibre
  ensemble.
\newblock {\em Journal of Physics A: Mathematical and Theoretical},
  41(37):375003, 2008.

\bibitem{BS}
A.~Borodin and C.D. Sinclair.
\newblock The {G}inibre ensemble of real random matrices and its scaling
  limits.
\newblock {\em Communications in Mathematical Physics}, 291:177--224, 2009.

\bibitem{Sommers2007}
H.-J. Sommers.
\newblock Symplectic structure of the real {G}inibre ensemble.
\newblock {\em Journal of Physics A: Mathematical and Theoretical},
  40(29):F671, 2007.

\bibitem{SW}
H.-J. Sommers and W.~Wieczorek.
\newblock General eigenvalue correlations for the real {G}inibre ensemble.
\newblock {\em Journal of Physics A: Mathematical and Theoretical},
  41(40):405003, 2008.

\bibitem{Akemann2010b}
G.~Akemann, M.~J. Phillips, and H.-J. Sommers.
\newblock The chiral gaussian two-matrix ensemble of real asymmetric matrices.
\newblock {\em Journal of Physics A: Mathematical and Theoretical},
  43(8):085211, 2010.

\bibitem{Osborn2004}
J.~C. Osborn.
\newblock Universal results from an alternate random-matrix model for {Q}{C}{D}
  with a baryon chemical potential.
\newblock {\em Phys. Rev. Lett.}, 93(22):222001, Nov 2004.

\bibitem{Akemann2005b}
G.~Akemann.
\newblock {The complex {L}aguerre symplectic ensemble of non-Hermitian
  matrices}.
\newblock {\em Nucl. Phys.}, B730:253--299, 2005.

\bibitem{Forrester2011}
P.~J. Forrester and A.~Mays.
\newblock Pfaffian point process for the {G}aussian real generalised eigenvalue
  problem.
\newblock {\em Probability Theory and Related Fields}, pages 1--47, 2011.
\newblock 10.1007/s00440-011-0361-8.

\bibitem{KSZ10}
B.A. Khoruzhenko, H.-J. Sommers, and K.~{\.Z}yczkowski.
\newblock Truncations of random orthogonal matrices.
\newblock {\em Phys. Rev. E}, 82(4):040106(R), 2010.

\bibitem{FS03}
Y.V. Fyodorov and H.-J. Sommers.
\newblock Random matrices close to {H}ermitian or unitary: overview of methods
  and results.
\newblock {\em Journal of Physics A: Mathematical and General},
  36(12):3303--3347, 2003.

\bibitem{KS}
B.~A. Khoruzhenko and H.-J. Sommers.
\newblock Non-hermitian random matrix ensembles.
\newblock {\em Oxford Handbook of Random Matrix Theory}, page 376, 2011.

\bibitem{Me04}
M.L. Mehta.
\newblock {\em Random Matrices, 3rd ed.}
\newblock Academic Press, 2004.

\bibitem{Gi84}
V.~L. Girko.
\newblock Circular law.
\newblock {\em Theory of Probability and Its Applications}, 29(4):694--706,
  1985.

\bibitem{Ba97}
Z.~D. Bai.
\newblock Circular law.
\newblock {\em Annals of Probability}, 25(1):494--529, 1997.

\bibitem{GT07}
F.~G{\"o}tze and A.~Tikhomirov.
\newblock The circular law for random matrices.
\newblock {\em Annals of Probability}, 38(4):1444--1491, 2010.

\bibitem{TV08}
T.~Tao and V.~Vu.
\newblock Random matrices: The circular law.
\newblock {\em Communications in Contemporary Mathematics}, 10(2):261--307,
  2008.

\bibitem{FZ97}
J.~Feinberg and A.~Zee.
\newblock Non-gaussian non-hermitian random matrix theory: Phase transition and
  addition formalism.
\newblock {\em Nuclear Physics B}, 501(3):643--669, 1997.

\bibitem{FSZ01}
J.~Feinberg, R.~Scalettar, and A.~Zee.
\newblock "single ring theorem" and the disk-annulus phase transition.
\newblock {\em Journal of Mathematical Physics}, 42(12):5718--5740, 2001.

\bibitem{GKZ09}
A.~Guionnet, M.~Krishnapur, and O.~Zeitouni.
\newblock The single ring theorem.
\newblock {\em Annals of Mathematics}, 174:1189--1217, 2011.

\bibitem{Wei08}
Y.~Wei and Y.~V. Fyodorov.
\newblock On the mean density of complex eigenvalues for an ensemble of random
  matrices with prescribed singular values.
\newblock {\em Journal of Physics A - Mathematical and Theoretical},
  41(50):335102, 2008.

\bibitem{Bo10}
E.~Bogomolny.
\newblock Asymptotic mean density of sub-unitary ensembles.
\newblock {\em Journal of Physics A - Mathematical and Theoretical},
  43(33):335102, 2010.

\bibitem{Forrester2010b}
P.~J. Forrester.
\newblock {\em Log-gases and Random Matrices}.
\newblock Princeton University Press, 2010.

\bibitem{Burda2010}
Z.~Burda, R.~A. Janik, and B.~Waclaw.
\newblock Spectrum of the product of independent random gaussian matrices.
\newblock {\em Phys. Rev. E}, 81(4):041132, Apr 2010.

\bibitem{Burda2010a}
Z.~Burda, A.~Jarosz, G.~Livan, M.~A. Nowak, and A.~Swiech.
\newblock Eigenvalues and singular values of products of rectangular gaussian
  random matrices.
\newblock {\em Phys. Rev. E}, 82(6):061114, Dec 2010.

\bibitem{Jarosz2010}
A.~Jarosz.
\newblock Addition of free unitary random matrices {I}{I}.
\newblock {\em arXiv: 1010.5220v1 [math-ph]}, 2010.

\bibitem{ZS01}
K.~{\.Z}yczkowski and H.~J. Sommers.
\newblock Induced measures in the space of mixed quantum states.
\newblock {\em Journal of Physics A: Mathematical and General}, 34:7111--7125,
  2001.

\bibitem{ZS00}
K.~{\.Z}yczkowski and H.-J. Sommers.
\newblock Truncations of random unitary matrices.
\newblock {\em Journal of Physics A: Mathematical and General}, 33(10):2045,
  2000.

\bibitem{Forrester2006}
P.~J. Forrester.
\newblock Quantum conductance problems and the {J}acobi ensemble.
\newblock {\em Journal of Physics A: Mathematical and General}, 39(22):6861,
  2006.

\bibitem{JF}
J.~Fischmann.
\newblock {\em PhD thesis in preparation}.
\newblock PhD thesis, Queen Mary University of London.

\bibitem{AS}
M.~Abramowitz and I.~A. Stegun, editors.
\newblock {\em Handbook of Mathematical Functions with Formulas, Graphs, and
  Mathematical Tables}.
\newblock Dover Publications, New York, 1972.

\bibitem{Forrester1999}
P.~J. Forrester and G.~Honner.
\newblock Exact statistical properties of the zeros of complex random
  polynomials.
\newblock {\em Journal of Physics A: Mathematical and General}, 32(16):2961,
  1999.

\bibitem{Akemann2001b}
G.~Akemann.
\newblock {Microscopic correlations of non-Hermitian {D}irac operators in
  three-dimensional {Q}{C}{D}}.
\newblock {\em Phys. Rev.}, D64:114021, 2001.

\bibitem{Akemann2009}
G.~Akemann, M.J. Phillips, and L.~Shifrin.
\newblock Gap probabilities in non-hermitian random matrix theory.
\newblock {\em Journal of Mathematical Physics}, 50:063504, 2009.

\bibitem{WK}
Y.~Wei and B.~A. Khoruzhenko.
\newblock O(n) colour-flavour transformations and characteristic polynomials of
  real random matrices.
\newblock {\em arXiv:0901.0746v1 [math-ph]}, 2009.

\bibitem{FK2007}
Y.~V. Fyodorov and B.~A. Khoruzhenko.
\newblock Averages of spectral determinants and `single ring theorem' of
  feinberg and zee.
\newblock {\em Acta Physica Polonica}, B38:4067--4078, 2007.

\bibitem{A}
G.~Akemann, M.~Kieburg, and M.~J. Phillips.
\newblock Skew-orthogonal {L}aguerre polynomials for chiral real asymmetric
  random matrices.
\newblock {\em Journal of Physics A -Mathematical and Theoretical}, 43(37),
  2010.

\bibitem{Edelman1994}
Alan Edelman, Eric Kostlan, and Michael Shub.
\newblock How many eigenvalues of a random matrix are real?
\newblock {\em Journal of the American Mathematical Society}, 7(1):247--267,
  1994.

\bibitem{BZ06}
I.~Bengtsson and K.~{\.Z}yczkowski.
\newblock {\em Geometry of Quantum States}.
\newblock Cambridge University Press, Cambridge, 2006.

\bibitem{BCSZ09}
W.~Bruzda, V.~Cappellini, H.-J. Sommers, and K.~{\.Z}yczkowski.
\newblock Random quantum operations.
\newblock {\em Physics Letters A}, 373(3):320--324, 2009.

\bibitem{RSZ11}
W.~Roga, M.~Smaczy{\'n}ski, and K.~{\.Z}yczkowski.
\newblock Composition of quantum operations and products of random matrices.
\newblock {\em Acta Phys. Pol.}, B 42:1123--1140, 2011.

\bibitem{Ho05}
A.~S. Holevo.
\newblock Complementary channels and the additivity problem.
\newblock {\em Theory of Probability and Its Applications}, 51(1):92--100,
  2006.

\bibitem{king}
C.~King, K.~Matsumoto, M.~Nathanson, and M.~B. Ruskai.
\newblock Properties of conjugate channels with applications to additivity and
  multiplicativity.
\newblock {\em Markov Processes and Related Fields}, 13:391 -- 423, 2007.

\bibitem{GZ10}
A.~Guionnet, , and O.~Zeitouni.
\newblock Support convergence in the single ring theorem.
\newblock {\em arXiv:1012.2624}, 2010.

\end{thebibliography}

\end{document}